%% file: main.tex
\documentclass[11pt]{article}
\usepackage[utf8]{inputenc}
\usepackage{amsmath}
\usepackage{amssymb}
\usepackage{amsthm}
\usepackage{mathtools}
\usepackage{xcolor}

\usepackage{multirow}
\usepackage{tabularx}%
\usepackage{enumitem}
\usepackage{booktabs}
\usepackage{caption, subcaption}

\usepackage[colorlinks]{hyperref}
\usepackage{authblk}
\usepackage[normalem]{ulem} 

\title{When does Metropolized Hamiltonian Monte Carlo provably outperform Metropolis-adjusted Langevin algorithm?}
\author[1]{Yuansi Chen}
\author[2]{Khashayar Gatmiry}
\author[3]{Minhui Jiang}
\affil[1,3]{ETH Z\"urich}
\affil[2]{MIT}
\date{}

\input{macros.tex}

\begin{document}

\maketitle

\begin{abstract}
  We analyze the mixing time of Metropolized Hamiltonian Monte Carlo (HMC) with the leapfrog integrator to sample from a distribution on $\real^\dims$ whose log-density is smooth, has Lipschitz Hessian in Frobenius norm and satisfies isoperimetry. We bound the gradient complexity to reach $\epsilon$ error in total variation distance from a warm start by $\tilde O(d^{1/4}\polylog(1/\epsilon))$ and demonstrate the benefit of choosing the number of leapfrog steps to be larger than 1. To surpass the previous analysis on Metropolis-adjusted Langevin algorithm (MALA) that has $\tilde{O}(d^{1/2}\polylog(1/\epsilon))$ dimension dependency~\cite{wu2021minimax}, we reveal a key feature in our proof that the joint distribution of the location and velocity variables of the discretization of the continuous HMC dynamics stays approximately invariant. This key feature, when shown via induction over the number of leapfrog steps, enables us to obtain estimates on moments of various quantities that appear in the acceptance rate control of Metropolized HMC. Notably, our analysis does not require log-concavity or independence of the marginals, and only relies on an isoperimetric inequality. To illustrate the relevance of the Lipschitz Hessian in Frobenius norm assumption, several examples that fall into our framework are discussed. 
\end{abstract}

\section{Introduction}

The Hamiltonian Monte Carlo (HMC) algorithm is a popular Markov chain Monte Carlo (MCMC) algorithm for sampling from a smooth distribution on $\real^\dims$. The algorithm originates from the physics literature~\cite{alder1959studies,duane1987hybrid} and has important applications in statistical
physics, computational chemistry, etc. Since its introduction to the statistics by Neal~\cite{neal1994improved}, HMC has since become very popular as a generic tool for sampling smooth distributions in Bayesian statistical inference. Notably, HMC is implemented the default sampling algorithm for sampling complex smooth distribution in many state-of-the-art software packages, including Stan~\cite{carpenter2017stan}, PyMC3~\cite{salvatier2016probabilistic}, Mamba and Tensorflow~\cite{abadi2016tensorflow}. 

HMC's popularity can be attributed to its superior empirical performance when compared to traditional MCMC sampling algorithms like Metropolized Random Walk (MRW) and Metropolis-adjusted Langevin Algorithm (MALA). There are several heuristic arguments that suggest why HMC might perform better. For example, based on intuition from the continuous-time limit of the Hamiltonian dynamics, it was suggested that HMC can suppress random walk behavior using momentum~\cite{neal2011mcmc}. Others highlighted HMC's ability to maintain a
constant asymptotic accept-reject rate with large step-size~\cite{creutz1988global}. Although these intuitive arguments may be compelling, their quantitative understanding is limited, and they cannot be relied on to determine the optimal free parameters of HMC. 

There are two fundamental quantitative questions that are crucial to understanding the fast convergence of HMC. The first question concerns the conditions under which HMC will converge faster than MALA. While HMC is expected to perform better, the specific conditions of this advantage are not well-understood. This lack of clarity poses difficulties when one has to choose between MALA and HMC.  The second question is related to the sensitivity of HMC to its free parameters, such as the step-size and number of integration steps. If too few integration steps are taken, HMC can behave similarly to MALA, whereas too many steps can cause HMC to wander back to the initial state's neighborhood, resulting in slower mixing and wasted computation~\cite{betancourt2014optimizing}. Although analysis over simplified product distributions exists~\cite{beskos2013optimal}, it remains unclear how to achieve the optimal free parameter choices under general settings. 

In the paper, we investigate the mixing time of Metropolized HMC for sampling a smooth target density satisfying an isoperimetric inequality in $\real^\dims$. With the help of the theoretical mixing time analysis, we aim to offer guidance over the optimal free parameter specification of HMC.  

\subsection{Related work}
We start by reviewing the existing literature on the Metropolized Adjusted Langevin Algorithm (MALA). MALA is essentially a special case of Metropolized HMC, where only one leapfrog integrator step is taken.  The unadjusted version of MALA, the unadjusted Langevin algorithm (ULA), has been extensively studied in the literature, with numerous papers discussing its properties~\cite{parisi1981correlation, grenander1994representations, dalalyan2017theoretical, durmus2017nonasymptotic, cheng2018convergence, durmus2019high, vempala2019rapid,altschuler_resolving_2023,bou-rabee_mixing_2023}. However, unadjusted sampling algorithms suffer from a significant limitation: their mixing time depends polynomially on the error tolerance. As a result, they are not suitable to obtain high-quality samples.

MALA was introduced by Besag~\cite{besag1994comments} as the adjusted Langevin algorithm that uses the Metropolis-Hastings step to ensure its convergence to the correct stationary distribution. Unlike ULA, MALA has a logarithmic dependence on the error tolerance~\cite{dwivedi2018log}. Extensive work has been done on establishing the mixing time of MALA for various distributions~\cite{roberts1996geometric,roberts1998optimal,bou2013nonasymptotic,dwivedi2018log,chen2020fast,lee2021structured,chewi2021optimal,wu2021minimax}. Specifically, the mixing time of MALA for sampling a $\mparam$-strongly log-concave distribution with second-order smoothness parameter $\Lparam$ is settled at $\tilde{O}(\dims^{\frac12} \Lparam/\mparam)$ from a warm start. This bound is shown to be tight up to logarithmic factors~\cite{wu2021minimax}. However, for general distributions with higher-order smoothness, it remains an open question how the mixing time scales with the dimension $\dims$. Roberts and Rosenthal~\cite{roberts1998optimal} suggested a $\dims^{\frac13}$ scaling limit for sufficiently regular product distributions, but more research is needed to extend these results beyond product distributions. 

While the mixing time of the Metropolized Adjusted Langevin Algorithm (MALA) has been well studied in recent years, much less is known about MCMC methods that rely on Hamiltonian Ordinary Differential Equations (ODEs). To transform HMC's continuous ODE into an algorithm, numerical integrators are often used to discretize the ODE and produce a Markov chain based on a finite number of gradient evaluations. The leapfrog or St\"ormer-Verlet integrator is a commonly used choice that performs well in practice, and it discretizes the integral of the force using a trapezoidal rule~\cite{neal2011mcmc}. 

For exact HMC with zero numerical error, \cite{mangoubi2017rapid} used a coupling argument to derive Wasserstein-1 distance convergence for strongly logconcave target distributions, which was sharpened in~\cite{chen_optimal_2022} and~\cite{bou2020coupling}. 

For unadjusted Hamiltonian Monte Carlo (uHMC), \cite{gouraud_hmc_2025} demonstrated that the leapfrog integrator has a gradient complexity of $\tilde{O}(d^{1/2}\epsilon^{-1/2})$ under the Lipschitz property of the Hessian. The result is obtained by examining a more general family of algorithms that are discretization of underdamped Langevin diffusions (ULD). Shen et al.~\cite{shen2019randomized} incorporate the idea of Randomized Midpoint Method (RMM) to ULD, obtained a gradient complexity of $\tilde O(d^{1/3}\epsilon^{-2/3})$. This dimension dependency is shown to tight in Cao et al.~\cite{cao2020complexity}. Recently, \cite{bou-rabee_unadjusted_2025} used this idea of RMM to propose a random-time integrator, which they call ``Stratified Monte Carlo'' for uHMC. This algorithm achieves a gradient complexity of $\tilde O(d^{1/3}\epsilon^{-2/3})$ without high-order derivative assumptions. All above results for uHMC are stated for convergence in the Wasserstein distance. Explicit total variation (TV) bounds for uHMC are scarce, \cite{bou-rabee_mixing_2023} showed that $O(d^{3/4}\epsilon^{-1/2} \log(d/\epsilon))$ gradient evaluations suffices for $\epsilon$-TV distance accuracy. 

Similar to the situation in Langevin algorithms, adding a Metropolis-Hastings step to uHMC can lead to a mixing time with logarithmic dependence on the error tolerance. In particular, Chen et al.~\cite{chen2020fast} demonstrate a gradient complexity of $\tilde{O}(d^{11/12}\polylog(1/\tole))$ for sampling from a strongly log-concave and third-order smooth target distribution in total variation (TV) distance under a warm start. It is worth noting that there is a gap between the above result and what is known for sampling a regular product distribution via Metropolized HMC, where the gradient complexity only has dimension dependency $\dims^{\frac14}$~\cite{beskos2013optimal}.


\subsection{Our contribution}
In this paper, we investigate the gradient complexity (or the number of gradient evaluations) needed for Metropolized HMC with leapfrog integrators to sample a distribution on $\real^\dims$ whose log-density is smooth, has Lipschitz Hessian in Frobenius norm and satisfies an isoperimetric inequality, from a warm start. We establish that to achieve $\tole$ error in total variation distance, the gradient complexity is $\tilde O\parenth{\dims^{\frac14} \log(1/\epsilon)}$. Because of the novelty in acceptance rate control, this is the first time that a $\dims^{\frac14}$ bound for Metropolized HMC with leapfrog integrators for general distributions, in contrasts with many previous bounds with at least $\dims^{\frac12}$ dependency~\cite{chen2020fast,chewi2021optimal,wu2021minimax}.

Additionally, our bound reveals that with the additional Hessian Lipschitz assumption, it is beneficial to choose the number of leapfrog steps $K > 1$.  Specifically, choosing $K=1$, which corresponds to MALA, results in a mixing upper bound of order $\dims^{\frac37}$ under the same assumptions. 

Our results only require isoperimetry and can deal with target distributions without assuming log-concavity. To compare with the existing work in the log-concave sampling setting, we translate our results and summarize them high accuracy sampling in Table~\ref{tab:1}. In the log-concave sampling setting, we assume that the target density is $\mparam$-strongly log-concave and $\Lparam$-log-smooth. Our results require an additional assumption that the target density is $\SSC \Lparam^{\frac32}$-Hessian Lipschitz in Frobenius norm (also called strongly Hessian Lipschitz), with $\SSC$ assumed to be constant for simplicity. In this setting, the target density satisfies an isoperimetric inequality with coefficient $\mparam^{\frac12}$.
\begin{table}[ht]
\centerline{
\begin{tabular}{ c c c c}
 \specialrule{.1em}{.05em}{.05em}
 HMC (\text{best} $K > 1$) & Initialization & Extra assumption & \#Gradient Evals \\\hline
\cite{beskos2013optimal} & warm & product distribution & $\dims^{\frac14}\dagger$ \\\hline
  \cite{chen2020fast} & $\Normal(x^*,\Lparam^{-1}\Ind_\dims)$ & Hessian Lipschitz & $\dims^{\frac{11}{12}} \kappa $  \\\hline
 this work & warm  & strongly Hessian Lipschitz & $\dims^{\frac{1}{4}}\kappa$ \\
\specialrule{.1em}{.05em}{.05em}
\specialrule{.1em}{.05em}{.05em}
MALA ($K=1$) & - & - & - \\\hline
\cite{roberts1998optimal} & warm & product distribution & $\dims^{\frac13}\dagger$ \\\hline
\cite{dwivedi2018log}& \multirow{ 2}{*}{warm / $\Normal(x^*,\Lparam^{-1}\Ind_\dims)$} & \multirow{ 2}{*}{N/A} &\multirow{ 2}{*}{ $\max\{d^{\frac12}\condi^{\frac32}, \dims\condi\}$}\\
 \cite{chen2020fast}& &\\\hline
 \cite{lee2020logsmooth} & $\Normal(x^*,\Lparam^{-1}\Ind_\dims)$ & N/A & $\dims \kappa $ 
 \\\hline
 \cite{chewi2021optimal}& warm & N/A & $ \dims^{\frac12}\condi^{\frac32}$\\\hline
 \cite{wu2021minimax}& warm & N/A & $ \dims^{\frac12}\condi$\\\hline
 \cite{altschuler_faster_2024}& $\Normal(x^*,\Lparam^{-1}\Ind_\dims)$ & N/A & $ \dims^{\frac12}\condi$\\\hline
  this work & warm  & strongly Hessian Lipschitz & $ \dims^{\frac37}\condi$\\
 \specialrule{.1em}{.05em}{.05em}
  \end{tabular}}
  \caption{Summary of $\tole$-mixing time in TV distance of Metropolized HMC and MALA for sampling a $\Lparam$-log-smooth and $\mparam$-strongly log-concave target under a warm start or a Gaussian initialization $\Normal(x^*,{\Lparam}^{-1}\Ind_\dims)$, where $x^*$ denotes the unique mode of the target density.  In the results with extra assumptions, the Hessian Lipschitz or strongly Hessian Lipschitz coefficients are considered constant. The MALA results can be considered as HMC results with the number of leapfrog steps $K$ chosen to be $1$. These statements hide constants and logarithmic factors in $\dims,\tole^{-1}$ or $\condi =\Lparam/\mparam$. $\dagger$ The $\kappa$ dependency is unknown and higher-order derivatives are assumed in~\cite{roberts1998optimal} and~\cite{beskos2013optimal}. }
  \label{tab:1}
\end{table}

The rest of paper is structured as follows. Section~\ref{sec:preliminaries} provides background information on Markov chains, the basics of Hamiltonian Monte Carlo and our main assumption on the target density. In Section~\ref{sec:main_results}, we present our main result, followed by a proof that assumes the proposal overlap bound and the acceptance rate control. Section \ref{sec:proposal_overlap} bounds the proposal overlap. In Section \ref{sec:acceptance}, we highlight our novel proof techniques to control the acceptance rate control.
Finally, in Section~\ref{sec:examples}, we provide examples of target
densities that satisfy our assumptions, and we discuss the implications of our findings.

\section{Preliminaries}
\label{sec:preliminaries}
In this section, we introduce the Metropolized HMC algorithm and provide necessary background knowledge to establish mixing time of Markov chains. 
\subsection{Markov chain basics}
We consider the problem of sampling from a target measure $\mu$ with density with respect to the Lebesgue measure on $\real^\dims$. Given a \textit{Markov chain with transition kernel} $P: \real^\dims \times \borel(\real^\dims) \to \real_{\geq 0}$ where $\borel(\real^\dims)$ denotes the Borel $\sigma$-algebra on $\real^\dims$, the $k$-step transition kernel $P^k$ is defined recursively by $P^k(x, dy) \defn \int_{z \in \real^\dims} P^{k-1}(x, dz) P(z, dy)$. 
The Markov chain is called \textit{reversible} with respect to the target measure $\mu$ if 
\begin{align*}
  \mu(dx) P(x, dy) =  \mu(dy) P(y, dx).
\end{align*}
One can associate the Markov chain with a transition operator $\transition_P$.
\begin{align*}
  \transition_P(\nu) (S) \defn \int_{y \in \real^\dims } d\nu(y) P(y, S), \quad \forall S\in \borel(\real^\dims).
\end{align*}
In words, when $\nu$ is the distribution of the current state, $\transition_P(\nu)$ gives the distribution of the next state. And $\transition_P^k(\nu) \defn \transition_{P^k}(\nu)$ is the distribution of the state after $k$ steps. 

\paragraph{$s$-conductance.} For $s \in (0, 1)$, we define the $s$-conductance $\Phi_s$ of the Markov chain $P$ with its stationary measure $\mu$ as follows
\begin{align}
  \label{eq:def_s_conductance}
  \Phi_s (P) \defn \inf_{S: s < \mu(S) < 1-s} \frac{\int_S P(x, S^c) \mu(dx)}{\min\braces{\mu(S), \mu(S^c)} - s}.
\end{align}
When compared to conductance (the case $s=0$), $s$-conductance allows us to ignore small parts of the distribution where the conductance is difficult to bound. 

\paragraph{Lazy chain.} Given a Markov chain with transition kernel $P$. We define its lazy variant $P^{\text{lazy}}$, which stays in the same state with probability at least $\frac{1}{2}$, as
\begin{align*}
  P^{\text{lazy}}(x, S) \defn \frac{1}{2} \delta_{x \to S} + \frac{1}{2} P(x,S).
\end{align*}
Here $\delta_{x \to \cdot}$ is the Dirac measure at $x$. 
Since the lazy variant only slows down the convergence rate by a constant factor, we study lazy Markov chains in this paper to simplify our
theoretical analysis.

\paragraph{Total variation distance.} Let the total variation (TV) distance between two probability distributions $\mathcal{P}_1, \mathcal{P}_2$ be 
\begin{align*}
  \tvdis(\mathcal{P}_1, \mathcal{P}_2) \defn \sup_{A \in \borel(\real^\dims)} \abss{\mathcal{P}_1(A) - \mathcal{P}_2(A)},
\end{align*}
where $\borel(\real^\dims)$ is the Borel sigma-algebra on $\real^\dims$. If $\mathcal{P}_1$ and $\mathcal{P}_2$ admit densities $p_1$ and $p_2$ respectively, we may write $\tvdis(\mathcal{P}_1, \mathcal{P}_2) = \frac{1}{2} \int \abss{p_1(x) - p_2(x)} dx$.

\paragraph{Warm start.} We say an initial measure $\mu_{0}$ is \textit{$\warm$-warm} if it satisfies
\begin{align*}
  \sup_{S \in \borel(\real^\dims)} \frac{\mu_{0}(S)}{\mu(S)} \leq \warm. 
\end{align*}

\paragraph{Mixing time.} For an error tolerance $\epsilon \in (0, 1)$, the total variation distance $\epsilon$-mixing time of the Markov chain $P$ with initial distribution $\mu_{0}$ and target distribution $\mu$ is defined as
\begin{align*}
  \tmix^P(\epsilon, \mu_{0}, \mu) \defn \inf\braces{k \in \naturalnum \mid \tvdis\parenth{\mathcal{T}^k_P (\mu_{0}), \mu} \leq \epsilon }. 
\end{align*}
The \textit{gradient complexity} of the Markov chain $P$ is defined as the product of the number of gradient evaluations per-iteration and the mixing time. 

\subsection{HMC basics}
In accordance with~\cite{neal2011mcmc}, we introduce a few basic facts about the HMC algorithm.
\paragraph{Continuous HMC dynamics.} The continuous HMC dynamics describes the evolution of a state vector $\cq_t \in \real^d$ and its momentum $\cp_t \in \real^d$ over time $t \geq 0$ based on a Hamiltonian function $\Ham: \real^\dims \times \real^\dims \to \real$ via Hamilton's equations:
\begin{align*}
  \frac{d \cq_t}{dt} &= \frac{\partial \Ham}{\partial p}(\cq_t, \cp_t)\\
  \frac{d \cp_t}{dt} &= -\frac{\partial \Ham}{\partial q}(\cq_t, \cp_t).
\end{align*}
Using our choice of the Hamiltonian function $\Ham(q, p) \defn f(q) + \frac{1}{2} \vecnorm{p}{2}^2$, it becomes
\begin{align}
  \label{eq:continuous_HMC_ODE}
  \frac{d \cq_t}{dt} &= \cp_t \notag \\
  \frac{d \cp_t}{dt} &= -\gradf (\cq_t).
\end{align}
We can write Eq.~\eqref{eq:continuous_HMC_ODE} in its integral form as follows
\begin{align}
  \label{eq:continuous_HMC_form}
  \cq_{t} &= \cq_0 + \int_0^t \cp_{s} \, ds \notag \\
  \cp_{t} &= \cp_0 - \int_0^t \gradf (\cq_s) \, ds.
\end{align}
We denote the forward mapping from $(\cq_0, \cp_0)$ to $(\cq_t, \cp_t)$ as $\bF_t$. That is, $\bF_t(\cq_0, \cp_0) = (\cq_t, \cp_t)$.
\paragraph{Metropolized HMC with leapfrog integrator.} To implement continuous HMC practically, it needs to be approximated through time discretization.  We use the leapfrog or St\"ormer-Verlet integrator to discretize the continuous HMC dynamics. For integer $K \geq 1$ and step-size $\step$, HMC with leapfrog integrator is an iterative algorithm that works as follows. At each iteration, given $q_0 \in \real^d$ and an independent $p_0 \sim \Normal(0, \Ind_d)$, it runs the following updates for $k = 0, 1, \ldots, K-1$:
\begin{align*}
  p_{k+\frac{1}{2} } &= p_{k } - \frac{\step}{2} \gradf(q_{k }) \\
  q_{k+1 } &= q_{k } + \step p_{k+\frac12 } \\
  p_{k+1 } &= p_{k+\frac{1}{2} } - \frac{\step}{2} \gradf(q_{k+1 }).
\end{align*}
That is,
\begin{equation}\label{eq:leapfrog_HMC_recursion}
\begin{aligned}
  q_{k+1 } &= q_{k } + \step p_{k } - \frac{\step^2}{2} \gradf(q_{k }) \notag \\
  p_{k+1 } &= p_{k } - \frac{\step}{2} \gradf(q_{k }) - \frac{\step}{2} \gradf(q_{k+1 }).
\end{aligned}
\end{equation}
The forward mapping from $(q_0, p_0)$ to $(q_{k }, p_{k })$ is denoted $\F_k$, thus $(q_k,p_k)=\F_k(q_0, p_0)$. Since discretizing the dynamics generates discretization errors at each iteration, a Metropolis-Hastings adjustment is used to ensure the correct stationary measure. The Metropolized HMC accepts the proposal $(q_K, p_K)$ with acceptance probability
\begin{align}
  \label{eq:HMC_acceptance}
  \Accept_{K, \step}(q_0, p_0) &\defn \min\braces{1, \frac{\exp(-\Ham(q_{K }, p_{K }))}{\exp(-\Ham(q_0, p_0)) }} \notag \\
  &= \min \braces{1, \exp\parenth{-f(q_{K }) - \frac{1}{2} \vecnorm{p_{K}}{2}^2 + f(q_0) + \frac{1}{2} \vecnorm{p_0}{2}^2 }}.
\end{align}
The resulting algorithm is what we call the Metropolized HMC with step-size $\step$ and $K$ leapfrog steps. Throughout the paper, this algorithm is what we refer to when we simply say the HMC algorithm. We denote its Markov transition kernel by $\transition_{\taghmc(K, \step)}$. The dependence on $K$ and $\step$ is omitted if it is clear from the context. Note that MALA can be seen as a special case of Metropolized HMC with a single leapfrog step (K=1), as described in~\cite{neal2011mcmc}. 

\subsection{Notation}
For a vector $q \in \real^\dims$, we use $\enorm{q} \defn \parenth{\sum_{i=1}^\dims q_i^2}^{\frac12}$ to denote its Euclidean norm. Given an integer $\ell \geq 1$, we introduce the $\ell$-th norm of a random variable $\vq$ on $\real^\dims$ as 
\begin{align*}
  \matsnorm{\vq}{\ell} \defn \parenth{\Exs \enorm{\vq}^\ell}^{\frac{1}{\ell}}.
\end{align*}
This $\matsnorm{\cdot}{\ell}$ norm satisfies the triangle inequality
\begin{align*}
  \matsnorm{\vq_1 + \vq_2}{\ell} \leq \matsnorm{\vq_1}{\ell} + \matsnorm{\vq_2}{\ell}.
\end{align*}
Additionally, for a function of random variables $s \mapsto \vq_s$, by Jensen's inequality, we have
\begin{align}
  \label{eq:random_vector_norm_integral}
  \matsnorm{\int_0^t \vq_s ds}{\ell}^{\ell} &\leq t^{\ell-1} \int_0^t  \matsnorm{\vq_s}{\ell}^{\ell} ds \leq \parenth{t \sup_{s \in [0,t]} {\matsnorm{\vq_s}{\ell}}}^\ell
\end{align}

For a matrix $A \in \real^{\dims \times \dims}$, we denote its spectral norm by $\vecnorm{A}{2} \defn \sup_{\enorm{x} \neq 0} \frac{\enorm{Ax}}{\enorm{x}}$ and its Frobenius norm by $\vecnorm{A}{F} \defn \parenth{\sum_{i,j=1}^\dims A_{ij}^2}^{\frac12}$. 
For a function $f: \real^\dims \to \real$ that is thrice differentiable, we denote its first, second and third derivatives at $x$ respectively by $\gradf(x) \in \real^\dims$, $\hessf_x \in \real^{\dims \times \dims}$ and $\nabla^3 f_x \in \real^{\dims \times \dims \times \dims}$. Here
\begin{align*}
  \brackets{\gradf(x)}_i = \frac{\partial}{\partial x_i} f(x), \brackets{\hessf_x}_{ij} = \frac{\partial^2}{\partial x_i \partial x_j} f(x), \brackets{\nabla^3 f_x}_{ijk} = \frac{\partial^3}{\partial x_i \partial x_j \partial x_k} f(x). 
\end{align*}
We use $[d]$ to denote the set $\braces{1, 2, \ldots, d}$.
For $A \in \real^{\dims \times \dims \times \dims}$ a 3-tensor, we use the notation $A[x, y, z] = \sum_{(i, j, k) \in [d]^3} A_{ijk} x_i y_j z_k$. For integer $k \geq 1$, for disjoint nonempty subsets $I_1, \ldots, I_k$ of $\braces{1, 2, 3}$, we define the $I_1 \ldots I_k$-norm of $A$ as
\begin{align}
  \label{eq:def_multi_index_norm}
  \vecnorm{A}{I_1 \ldots I_k} \defn \sup \braces{\sum_{\mathbf{i} \in [d]^3 } A_{\mathbf{i}} x^{(1)}_{\mathbf{i}_{I_1}} \cdots x^{(k)}_{\mathbf{i}_{I_k}} \middle|  x^{(j)}_{\mathbf{i}_{I_j}} \in \real^{ \dims \times \abss{I_j}}, \sum_{\mathbf{i}_{I_j} \in  [d]^{\abss{I_j}} } \parenth{ x^{(j)}_{\mathbf{i}_{I_j}}}^2 \leq 1, \forall j \in [k]}.
\end{align}
Here $\mathbf{i}_{I_j} \in [d]^{\abss{I_j}}$ is the tuple formed by $\mathbf{i}_\ell$ for $\ell \in I_j$. In particular, for $k=1$, $I_1 = \braces{1, 2, 3}$, we obtain
\begin{align*}
  \vecnorm{A}{\braces{1,2,3}} = \parenth{\sum_{\mathbf{i} \in [d]^3} A_{\mathbf{i}}^2 }^{1/2}.
\end{align*}
As another example, for $k=3, I_1 = \braces{1}, I_2 = \braces{2}, I_3 = \braces{3}$, we obtain
\begin{align*}
  \vecnorm{A}{\braces{1}\braces{2}\braces{3}} = \sup\braces{\sum_{i,j,k = 1}^d A_{ijk} x_i y_j z_k \middle| \vecnorm{x}{2}, \vecnorm{y}{2}, \vecnorm{z}{2} \leq 1 }.
\end{align*}
For any partition $\mathcal{P}$ of $[3]$, we have $\vecnorm{A}{\braces{1}\braces{2}\braces{3}} \leq \vecnorm{A}{\mathcal{P}} \leq \vecnorm{A}{[3]}$.
Additionally, if the $I_1\cdots I_k$-norm of the third derivative of a function $f$ bounded for any $x \in \real^\dims$, denote $\vecnorm{\jerkf_\cdot}{I_1\ldots I_k} \defn \sup_{x \in \real^\dims} \vecnorm{\jerkf_x}{I_1\ldots I_k}$. The norms for higher order tensors can be defined similarly (see e.g.~\cite{latala2006estimates}). 

For $F: \real^\dims \times \real^\dims \to \real^\dims$ a bivariate vector-valued function, we denote its derivative with respect to its first variable as $\D_1 F$, for $F = \parenth{F_1, \ldots, F_\dims}$,
\begin{align*}
  \D_1 F (x, y) = \brackets{\frac{\partial F_i(x, y)}{\partial x_j}  }_{1\leq i, j\leq \dims}.
\end{align*}
Similarly, $\D_2 F$ is the derivative with respect to its second variable.

The big-O notation $O(\cdot)$ and big-Omega notation $\Omega(\cdot)$ are used to denote bounds ignoring constants. For example, we write $g_1(x)=O(g_2(x))$ if there exists a universal constant $c>0$ such that $g_1(x)\leq c g_2(x)$. Adding a tilde above these notations such as $\tilde{O}(\cdot)$ and $ \tilde{\Omega}(\cdot)$ denotes bounds ignoring logarithmic factors for all symbols. We write $a \lesssim b$ if there exists a universal constant $c>0$ such that $a \leq c b$. We use $\text{poly(x)}$ to denote a polynomial of $x$ and $\polylog(x)$ to denote a polynomial of $\log(x)$.

\subsection{Regularity properties of a target density}
Below we list the regularity properties that one could impose on a target density.\label{subsec:regularityprop} 
\begin{itemize}

  \item A twice differentiable function $f: \real^\dims \to \real$ is \textit{$\Lparam$-smooth} if $ \vecnorm{\hessf_x}{2} \leq \Lparam, \forall x \in \real^\dims$. A target density $\target \propto e^{-f}$ is called \textit{$\Lparam$-log-smooth} if $f$ is $\Lparam$-smooth. 
  \item A twice differentiable function $f: \real^\dims \to \real$ is \textit{$\mparam$-strongly convex} if $\hessf_x \succeq \mparam \Ind_\dims, \forall x \in \real^\dims$,
  \item A thrice differentiable function $f: \real^\dims \to \real$ is \textit{$\SSC \Lparam^{\frac32}$-strongly Hessian Lipschitz} if its third derivative has its $\braces{1,2}\braces{3}$-norm bounded by $\SSC \Lparam^{\frac32}$
  \begin{align*}
    \vecnorm{\jerkf_x}{\braces{1,2}\braces{3}} \leq \SSC \Lparam^{\frac32}, \forall x \in \real^\dims.
  \end{align*}
  The above assumption implies that $\vecnorm{  \jerkf_{x}[h, \cdot, \cdot] }{F} \leq \SSC \Lparam^{\frac32} \enorm{h}, \forall h \in \real^\dims$, and 
  \begin{align*}
    \vecnorm{\hessf_{x} - \hessf_{y} }{F} \leq \SSC \Lparam^{\frac32} \enorm{x - y}, \forall y \in \real^\dims. 
  \end{align*}
  The last equation is the main reason we named this assumption ``strongly Hessian Lipschitz'', as the use of Frobenius norm makes it a stronger assumption the usual Hessian Lipschitz assumption which uses the spectral norm. 
  \item A thrice differentiable function $f: \real^\dims \to \real$ is \textit{$\SSC$-strongly self-concordant}~\cite{laddha2020strong} if for any $x, h \in \real^d$, we have
  \begin{align*}
    \vecnorm{ \hessf_{x}^{-\frac12} \jerkf_{x}[h, \cdot, \cdot] \hessf_{x}^{-\frac12}}{F} \leq \SSC \vecnorm{h}{x},
  \end{align*}
  where $\vecnorm{h}{x}^2 \defn h\tp \hessf_x h$. It is not hard to see that $\SSC$-strongly self-concordance implies $\SSC \Lparam^{\frac32}$-strongly Hessian Lipschitz. 
  \item A target density $\target \sim e^{-f}$ satisfies the \textit{Cheeger's isoperimetric inequality} with coefficient $\isop_\target$ if for any partition $(S_1, S_2, S_3)$ of $\real^\dims$, the following inequality is satisfied
  \begin{align*}
    \target(S_3)\geq \isop_\target \cdot \dist(S_1,S_2)\cdot \target(S_1)\target(S_2),
  \end{align*}
  where the distance between sets $S_1$, $S_2$ is defined as $\dist(S_1,S_2)=\inf_{x\in S_1,y\in S_2}\enorm{x-y}$.
  \item A target density $\mu \propto e^{-f}$ is \textit{subexponential-tailed} if there exists $\lambda > 0$ such that
  \begin{align*}
    \lim_{\enorm{x} \to \infty} e^{-\lambda \enorm{x}} e^{-f(x)} \to 0.
  \end{align*}
  Under the $\Lparam$-log-smooth assumption, the above definition is equivalent to the sub-exponential random vector definition (see~\cite{zajkowski2019norms}).
  \begin{align*}
    \sup_{\vecnorm{v}{1} \leq 1 } \vecnorm{\angles{v, X}}{\psi_1} \text{ is finite for $X \sim e^{-f}$, }
  \end{align*}
  where $\vecnorm{X}{\psi_1} = \inf\braces{t > 0: \Exs \exp(\abss{X}/t) \leq 2}$. It can be seen as an extension of sub-exponential random variable in one dimension~\cite{vershynin2018high}. It is well known that if $\mu$ satisfies the Cheeger isoperimetric inequality, then it is also subexponential-tailed. More precisely, the Poincar\'e constant of $\mu$ follows from Cheeger's isoperimetric coefficient by Cheeger's inequality~\cite{cheeger2015lower,maz1960classes} (see also Theorem 1.1 in~\cite{milman2009role}), and the Poincar\'e constant of $\mu$ implies subexponential tail bounds following Gromov and Milman~\cite{gromov1983topological} (see also Theorem 1.7 of~\cite{gozlan2015dimension}). 
\end{itemize}

\begin{assumption}
  \label{ass:assumption_main}
  The target density $\target \propto e^{-f}$ on $\real^\dims$ satisfies that $f$ is $\Lparam$-smooth, $\SSC\Lparam^{\frac32}$-strongly Hessian Lipschitz. Additionally, it satisfies the Cheeger's isoperimetric inequality with coefficient $\isop_\target$.
\end{assumption}

\section{Main results}
\label{sec:main_results}

\begin{theorem}
  \label{thm:HMC_main}
  Let $\mu\propto e^{-f}$ be a target density on $\real^d$ that satisfies Assumption~\ref{ass:assumption_main}. For any error tolerance $\tole \in (0, 1)$, from any $M$-warm initial measure $\mu_0$, if the HMC parameter choices are such that 
  \begin{align*}
    K\step \Lparam^{\frac12} &\leq \frac{1}{ 4 (\SSC +1)^{\frac13} } \\ 
    K^{1+\frac1\ell} \parenth{ \step^3  \Lparam^{\frac32} \dl^{\frac12} +  \step^5 \Lparam^{\frac{5}{2}} \dl + \step^7 \Lparam^{\frac72} \dl^{\frac32}} &\leq \frac{1}{c \parenth{\SSC + 1} \ell^{\frac32} },
  \end{align*}
  where $\ell \geq 2\ceils{c' \log\parenth{\max\braces{1,\frac{1}{K\step\isop_\target}}\frac{M}{\tole}}}$, $\dl = \dims + 2(\ell-1)$ and $c, c'$ are universal constants,
  then the $\tole$-mixing time of lazy Metropolized HMC satisfies
  \begin{align*}
    \tmix^\taghmc(\tole, \mu_0, \target) = O\parenth{\frac{1}{K^2\step^2 \isop_\target^2} \log\parenth{\frac{M}{\tole}}}.
  \end{align*}
  Consequently, the gradient complexity is $O\parenth{\frac{1}{K\step^2 \isop_\target^2} \log\parenth{\frac{M}{\tole}}}$.
\end{theorem}
To make our result easier to understand, we state two corollaries for the best choice of $K, \step$ and another one for the choice $K=1$, which corresponds to the case of MALA. 
\begin{corollary}[Best HMC]
  \label{cor:best_HMC}
  Let $\mu\propto e^{-f}$ be a target density on $\real^d$ that satisfies Assumption~\ref{ass:assumption_main}. For any error tolerance $\tole \in (0, 1)$, from any $M$-warm initial measure $\mu_0$, if the HMC parameter choices are such that
  \begin{align*}
    \step^2 &\defn  \frac{c}{\Lparam \parenth{\dims + 2\ell}^{\frac{1}{2}} \parenth{\SSC + 1}^{\frac23} \ell^{\frac32}} \\
    K &\defn \step^{-1} \cdot \frac{c'}{2 \Lparam^{\frac12} \parenth{\SSC + 1}^{\frac13} }
  \end{align*}
  where we define $\ell\defn c''\max\braces{\log(\dims),\log\parenth{\frac{M} {\tole}},\log\parenth{\frac{\Lparam(\SSC+1)}{\isop_\target^2}\frac{M}{\tole}}}$ and $c,c',c''$ are universal constants, then the $\tole$-mixing time of lazy Metropolized HMC satisfies
  \begin{align*}
    \tmix^\taghmc(\tole, \mu_0, \target) = O\parenth{\frac{\Lparam \parenth{\SSC + 1}^{\frac23} }{\isop_\target^2} \log\parenth{\frac{M}{\tole}}}.
  \end{align*}
  Consequently, the gradient complexity is $O\parenth{\frac{\Lparam \parenth{\SSC + 1}^{\frac23}}{\isop_\target^2} \parenth{\dims + 2\ell}^{\frac14}\ell^{\frac34}\log\parenth{\frac{M}{\tole}}}$.
\end{corollary}

From a warm start, assume $\SSC$ is of order constant, then the best parameter choice of HMC ($\step \approx \Lparam^{-1} \dims^{-\frac14}, K \approx \dims^{\frac14}$) gives a mixing time of order $\dims^{\frac14}$, considering only the dimension dependency. In the log-concave sampling setting where the target is $\mparam$-strongly log-concave and $L$-log-smooth ~\cite{wu2021minimax}, the isoperimetric coefficient $\isop_\target = \log(2) \mparam^{\frac{1}{2}}$. Hence, the gradient complexity becomes $\tilde{O}(\dims^{\frac14} \kappa)$, whereas the rate in~\cite{wu2021minimax} is $\tilde{O}(\dims^{\frac12} \kappa)$ where $\kappa = \Lparam/\mparam$.

Remark that to obtain the $\dims^{\frac14}$ rate for the class of target densities in the corollary, the warm start assumption cannot be removed. Because without the warm start assumption, HMC gradient complexity is known to be lower bounded by $\dims^{\frac12}$ for a Gaussian target density~\cite{lee2021lower}. Additionally, if we don't assume $\SSC$ is of order constant and for example assume $\SSC$ is of order $\dims^{\frac12}$, then the $K$ and $\step$ choices in Corollary~\ref{cor:best_HMC} are no longer optimal. According to the bounds in ~\cite{wu2021minimax}, it becomes better to choose $K=1$ and get back to MALA. 

\begin{corollary}[MALA]
  \label{cor:MALA}
  Let $\mu\propto e^{-f}$ be a target density on $\real^d$ that satisfies Assumption~\ref{ass:assumption_main}. For any error tolerance $\tole \in (0, 1)$, from any $M$-warm initial measure $\mu_0$, if MALA step-size is chosen as
  \begin{align*}
    \step^2 &\defn \frac{c}{\Lparam(\SSC+1)^{\frac23} \parenth{\dims + 2\ell}^{\frac{3}{7}} \ell},
  \end{align*}
  where we choose $\ell\defn c'\max\braces{\log(\dims),\log\parenth{\frac{M}{\tole}},\log\parenth{\frac{\Lparam(\SSC+1)}{\isop_\target^2}\frac{M}{\tole}}}$ for universal constants $c,c'$, then the $\tole$-mixing time of MALA satisfies
  \begin{align*}
    \tmix^\tagmala(\tole, \mu_0, \target) = O\parenth{\frac{\Lparam(\SSC+1)^{\frac23} \parenth{\dims + 2\ell}^{\frac{3}{7}}\ell}{\isop_\target^2} \log\parenth{\frac{M}{\tole}}}.
  \end{align*}
\end{corollary}

Corollary~\ref{cor:MALA} shows that when $\SSC$ is of order constant, MALA has a mixing rate $\dims^{\frac{3}{7}}$ from a warm start which is better than the previous bound $\dims^{\frac12}$ in~\cite{wu2021minimax}. However, when $\SSC \geq \dims^{\frac14}$, our bound no longer shows an improvement.

We note that our results generalize to weaker smoothness assumptions. For example, \cite{bou-rabee_mixing_2023} assumed the injective norm of $\nabla^3 f$ is bounded by $L_H$, that is, 
\begin{equation}
 \matsnorm{\nabla^3 f}{\braces{1}\braces{2}\braces{3}}\leq L_H.  
\end{equation}
Under this assumption, we can fix $\SSC \defn \Lparam^{-\frac23}\dims^{\frac12} L_H$, then Corollary \ref{cor:best_HMC} implies a mixing time of $\Ot({\dims^{\frac{7}{12}}\text{polylog}(\frac{1}{\epsilon})})$. Compared to the mixing time of $O(d^{\frac34}\epsilon^{-\frac12}\log(\frac{\dims}{\epsilon}))$ for unadjusted HMC in \cite{bou-rabee_mixing_2023}, our bound is tighter under a warm start, and we do not require fourth-order smoothness assumption.

\paragraph{Main proof strategy.} We follow the framework of lower bounding the conductance of Markov chains to analyze mixing times~\cite{sinclair1989approximate,lovasz1993random}. The following lemma reduces the problem of mixing time analysis to lower bounding the $s$-conductance $\Phi_s$. 

\begin{lemma}[Lovasz and Simonovits~\cite{lovasz1993random}]
  \label{lem:lovasz_lemma}
  Consider a reversible lazy Markov chain with kernel $P$ and stationary measure $\target$. Let $\mu_{0}$ be an $M$-warm initial measure. Let $0 < s < \frac{1}{2}$. Then 
  \begin{align*}
    \tvdis\parenth{\transition_P^k (\mu_0), \target} \leq Ms + M \parenth{1 - \frac{\Phi_s^2}{2}}^k .
  \end{align*}
\end{lemma}
In order to lower bound the conductance in the high probability region, we follow the proof of Theorem~3 in~\cite{wu2021minimax}, which is an extension of a result by Lov\'asz~\cite{lovasz1999hit} (see also Lemma 2 in~\cite{dwivedi2018log}). Informally, it states that as long as $\tvdis\parenth{\transition_x, \transition_y} \leq 1-\rho$ whenever $\enorm{x-y} \leq \Delta$, we have $\Phi_s \geq \frac{\rho\Delta}{\psi_\mu}$. Consequently, it remains to bound the transition overlap, namely the TV-distance between $\transition_x$ and $\transition_y$ for two close points $x, y$.

We bound the TV-distance via the following triangle inequality,
\begin{align*}
  \tvdis\parenth{\transition_x, \transition_y} \leq \tvdis\parenth{\transition_x, \proposal_x} + \tvdis\parenth{\proposal_x, \proposal_y} + \tvdis\parenth{\proposal_y, \transition_y},
\end{align*}
where $\proposal_x$ is the proposal probability measure of HMC at $x$. 
The remaining part of the proof is naturally divided into two parts: in Section~\ref{sec:proposal_overlap}, we bound the proposal overlap, which is the TV-distance between $\proposal_x$ and $\proposal_y$ for two close points $x, y$; in Section~\ref{sec:acceptance}, we bound the acceptance rate in Eq.~\eqref{eq:HMC_acceptance}.

The following lemma controls the proposal overlap. Note that the statement is similar to Lemma 16 in \cite{bou-rabee_mixing_2023}. The difference is that we assumed $\vecnorm{\nabla^3 f}{\braces{12}\braces{3}}\leq \gamma L^{\frac32}$ while \cite{bou-rabee_mixing_2023} assumed $\vecnorm{\nabla^3 f}{\braces{1}\braces{2}\braces{3}}\leq L_H$. Since $\vecnorm{\nabla^3 f}{\braces{12}\braces{3}}\leq \sqrt{\dims} \vecnorm{\nabla^3 f}{\braces{1}\braces{2}\braces{3}}$, inserting $\gamma\defn  L^{-\frac32}\sqrt{\dims}L_H$ in Lemma \ref{lem:proposal_KL_div_bound} recovers Lemma 16 in \cite{bou-rabee_mixing_2023}.  
\begin{lemma}[Proposal overlap]
  \label{lem:proposal_KL_div_bound}
  Let the target density $\mu \propto e^{-f}$ be $\Lparam$-log-smooth and $\SSC\Lparam^{\frac32}$-strongly Hessian Lipschitz.
  For $\enorm{q_0-\tilq_0} \leq \frac{1}{64} K\step$, we have
  \begin{align*}
    \kldiv{\proposal_{q_0}}{\proposal_{\tilq_0}} \leq \frac{1}{64} + K^6 \step^6 \SSC^2 \Lparam^3. 
  \end{align*}
  By Pinsker's inequality, we also have
  \begin{align*}
    \tvdist{\proposal_{q_0}}{\proposal_{\tilq_0}} \leq \parenth{\frac{1}{32} +  2K^6 \step^6 \SSC^2 \Lparam^3}^{\frac12}.
  \end{align*}
\end{lemma}
The proof of Lemma~\ref{lem:proposal_KL_div_bound} is provided in Subsection~\ref{sub:proof_of_lem:proposal_KL_div_bound} after we have introduced all the intermediate results.

\begin{lemma}[Acceptance rate control]
  \label{lem:acceptance_lower_bound}
  Let the target density $\mu \propto e^{-f}$ be $\Lparam$-log-smooth, $\SSC\Lparam^{\frac32}$-strongly Hessian Lipschitz and subexponential-tailed. There exists a universal constant $c, c' > 0$ such that for any $\delta \in (0, 1)$, if we  fix any $\ell \geq 2\ceils{c' \log(1/\delta)}$, $\dl = \dims + 2(\ell-1)$ and $K, \step$ such that 
  \begin{align*}
    K^{1+\frac{1}{\ell}} \parenth{\SSC + 1} \parenth{ \step^3 \ell^{\frac32} \Lparam^{\frac32} \dl^{\frac12} +  \step^5 \ell^{\frac12} \Lparam^{\frac{5}{2}} \dl + \step^7 \Lparam^{\frac72} \dl^{\frac32}} \leq \frac{1}{c}.
  \end{align*}
  Then there exists a set $\Lambda \subset \real^{\dims} \times \real^{\dims}$ with $\Prob_{(q_0, p_0) \sim \mu \times \Normal(0, \Ind_\dims)}((q_0, p_0) \in \Lambda) \geq 1-\delta$, such that for $(q_0, p_0) \in \Lambda$ and $(q_{K }, p_{K }) = F_K(q_0, p_0)$ as defined in Eq~\eqref{eq:leapfrog_HMC_recursion}, we have
  \begin{align*}
    - f(q_{K }) - \frac{1}{2} \enorm{p_{K }}^2 + f(q_0) + \frac{1}{2} \enorm{p_0}^2 \geq  - \frac{1}{4}.
  \end{align*}
\end{lemma}
The proof of Lemma~\ref{lem:acceptance_lower_bound} is provided in Subsection~\ref{sub:proof_of_lem:acceptance_lower_bound} after we have introduced all the intermediate results about the acceptance rate in Section~\ref{sec:acceptance}.

\subsection{Proof of Theorem~\ref{thm:HMC_main}}
Given the novel proposal overlap bound in Lemma~\ref{lem:proposal_KL_div_bound} and the novel acceptance rate control in Lemma~\ref{lem:acceptance_lower_bound}, the proof to bound $s$-conductance and then mixing time is quite standard. For completeness, we adapt the proof of Theorem 3 in~\cite{wu2021minimax} (see also~\cite{dwivedi2018log,lee2020logsmooth}) to incorporate our new proposal overlap and acceptance rate bounds. 

Applying Lemma~\ref{lem:lovasz_lemma}, with the choice
  \begin{align*}
    s = \frac{\tole}{2M}, n \geq \frac{2}{\Phi_s^2} \log \frac{2M}{\tole},
  \end{align*}
  it is sufficient to lower bound the conductance $\Phi_s$. Let $\delta \in (0, 1)$ to be chosen later, and $\Lambda$ be the set in Lemma~\ref{lem:acceptance_lower_bound}. Define 
  \begin{align*}
    \Xi \defn \braces{q_0 \in \real^\dims: \Prob_{p_0 \sim \Normal(0, \Ind_\dims)} \parenth{(q_0, p_0)\in \Lambda} \geq \frac{15}{16} }. 
  \end{align*}
  Then $\Prob_{q_0 \sim \target}(q_0 \in \Xi) \geq 1 - 16 \delta$. Applying Lemma~\ref{lem:acceptance_lower_bound}, we obtain
  \begin{align*}
    \Accept(q_0, p_0) \geq e^{-\frac{1}{4}}, \text{ for } (q_0, p_0) \in \Lambda.
  \end{align*}
  Consequently, we have
  \begin{align*}
    \tvdis\parenth{\transition_{q_0}, \proposal_{q_0}} = 1 - \Exs_{p_0 \sim \Normal(0, \Ind_\dims)}[\Accept(q_0, p_0)] \leq 1 - \frac{15}{16} e^{-\frac14} \leq \frac{1}{3}, \forall q_0 \in \Xi. 
  \end{align*}
  Together with Lemma~\ref{lem:proposal_KL_div_bound}, under the condition $K \step \SSC^{\frac13}\Lparam^{\frac12} \leq \frac{1}{4}$, whenever $\enorm{x-y} \leq \frac{1}{64} K\step$, we have
  \begin{align}
    \label{eq:tv_dist_lazyxy}
    \tvdis\parenth{\transition_x^{\text{lazy}}, \transition_y^{\text{lazy}}} &\leq \frac{1}{2} + \frac{1}{2} \tvdis\parenth{\transition_x, \transition_y} \notag \\
    &\leq \frac{1}{2} + \frac{1}{2} \parenth{\frac{2}{3} + \frac{1}{4}} \notag \\
    &= 1 - \frac{1}{24},
  \end{align}
  where $\transition_x^{\text{lazy}}$ is the lazy version of $\transition_x$. To lower bound the $s$-conductance, consider any $S$ measurable set with $\target(S) \in (s, \frac12]$. Define the sets 
  \begin{align*}
    S_1 \defn \braces{x \in S \mid \transition_x^{\text{lazy}}(S^c) < \frac{1}{64}}, S_2 \defn \braces{x \in S^c \mid \transition_x^{\text{lazy}}(S) < \frac{1}{64}}, S_3 \defn (S_1 \cup S_2)^c. 
  \end{align*}
  Note that $S_1$ and $S_2$ are the sets that do no conduct well. Now there are two cases.
  \paragraph{If $\target(S_1) \leq \frac12 \target(S)$ or $\target(S_2) \leq \frac12 \target(S^c)$,}then the sets that conduct well are large enough:
  \begin{align*}
    \int_S \transition^{\text{lazy}}_x (S^c) \target(dx) \geq \frac{1}{128} \target(S).
  \end{align*}
  \paragraph{Otherwise,} for any $x \in S_1 \cap \Xi$ and $y \in S_2 \cap \Xi$, we have
  \begin{align*}
    \tvdis\parenth{\transition_x^{\text{lazy}}, \transition_y^{\text{lazy}}} \geq \abss{\transition_x^{\text{lazy}}(S) - \transition_y^{\text{lazy}}(S)} \geq \frac{63}{64} - \frac{1}{64} \geq \frac{23}{24}.
  \end{align*}
  Together with Eq.~\eqref{eq:tv_dist_lazyxy}, we obtain $\dist(S_1 \cap \Xi,S_2 \cap \Xi) \geq \frac{K\step}{64}$. The isoperimetric inequality on $\target$ leads to 
  \begin{align*}
    \target(S_3 \cup \Xi^c) \geq \frac{K\step}{64} \isop_\target \target(S_1 \cap \Xi) \target(S_2 \cap \Xi). 
  \end{align*}
  Since $\target(\Xi) \geq 1 - 16 \delta$, after cutting off the parts outside $\Xi$ under the condition 
  \begin{align}
    \label{eq:delta_condition}
    16\delta\leq \min\braces{\frac{1}{4}s, \frac{1}{64*64} K \step \isop_\target s},
  \end{align}
  we obtain that
  \begin{align*}
    \target(S_3) \geq c \step \isop_\target \target(S),
  \end{align*}
  for a universal constant $c > 0$. In this case, we have
  \begin{align*}
    \int_S \transition^{\text{lazy}}_x (S^c) \target(dx) \geq \frac{1}{128} \target(S_3) \geq \frac{c}{128*64} K \step \isop_\target \target(S).
  \end{align*}
  Combining the two cases, the $s$-conductance is lower bounded as follows
  \begin{align*}
    \Phi_s \geq \frac{1}{128} \min\braces{1, c' K\step \isop_\target}. 
  \end{align*}
  $\delta$ has to be chosen according to Eq.~\eqref{eq:delta_condition} and the step-size $\step$ has to be chosen according to $K^6 \step^6 \SSC^2 \Lparam^3 \leq \frac{1}{16}$ and the first equation in Lemma~\ref{lem:acceptance_lower_bound}.

\section{Proposal overlap}
\label{sec:proposal_overlap}
Given $K$ and $\step$, recall that $\proposal_x$ is the proposal probability measure of HMC at $x$. We are interested in the proposal overlap. That is, how close should $x$ and $y$ be so that the TV distance $\tvdist{\proposal_x}{\proposal_y}$ is smaller than 1.

For $j\in [K]$, let $\Fo_j : \real^{\dims} \times \real^{\dims} \to \real^\dims$ be the first output of $F_j$. That is, $\Fo_j(q_0, p_0) = q_{j }$ as defined in Eq.~\eqref{eq:leapfrog_HMC_recursion}. For a fixed $q_0 \in \real^\dims$, the distribution of $q_{j }$ given $q_0$ is the push-forward of the standard Gaussian distribution through the map $x \mapsto \Fo_j(q_0, x)$. If this map is invertible, we can introduce the inverse map $x \mapsto \Go_j(q_0, x)$. Let $\Fo \defn \Fo_K$ and $\Go \defn \Go_K$. Then we obtain an explicit expression of the density $\rho_{q_0}$ for the distribution of $q_{K }$, 
\begin{align}
  \label{eq:density_proposal}
  \rho_{q_0}(y) = \varphi(\Go(q_0, y)) \det\parenth{\D_2 \Go(q_0, y)},
\end{align}
where $\varphi$ is the density of the standard Gaussian and $\D_2 \Go$ is the derivative of $\Go$ with respect to its second variable. 

For two initial points $q_0, \tilq_0 \in \real^\dims$, the TV distance between their proposal distributions is
\begin{align*}
  \tvdist{\proposal_{q_0}}{\proposal_{\tilq_0}} = \frac{1}{2} \int_{x \in \real^{\dims}} \abss{\rho_{q_0}(x) - \rho_{\tilq_0}(x) } dx.
\end{align*}
We bound the TV distance via KL divergence, after applying Pinsker's inequality. Using the expression of the proposal density in Eq.~\eqref{eq:density_proposal}, the KL divergence is as follows
\begin{align}
  \label{eq:KL_div_first_step}
  &\quad \kldiv{\proposal_{q_0}}{\proposal_{\tilq_0}} \notag \\
  &= \int \rho_{q_0}(y) \log\parenth{\frac{\rho_{q_0}(y)}{\rho_{\tilde{q}_0}(y) } } dy \notag \\
  &= \int \rho_{q_0}(y) \brackets{\log\parenth{\frac{\varphi(\Go(q_0, y))}{\varphi(\Go(\tilq_0, y))} } + \log\det \D_2 \Go(q_0, y) - \log\det \D_2 \Go(\tilq_0, y) } dy \notag \\
  &=  \int \rho_{q_0}(y) \brackets{ \frac{1}{2} \parenth{- \enorm{\Go(q_0, y)}^2 + \enorm{\Go(\tilq_0, y) }^2} + \log \det \D_2 \Go(q_0, y) - \log\det \D_2 \Go(\tilq_0, y)} dy.
\end{align}

\subsection{Intermediate results for the proposal overlap}
Looking at Eq.~\eqref{eq:KL_div_first_step}, we need to bound the derivatives of $\Fo$ before we provide a bound on the KL divergence. We state these bounds about $\Fo$ in this section. Their proofs are deferred to Appendix~\ref{app:add_proofs_proposal}. 

Lemma~\ref{lem:J_F_p_bound_better} ensures that $\Fo(q, \cdot)$ is invertible for any $q$. It is essentially a reformulation of Lemma~7 in~\cite{chen2020fast} in this paper's notation.
\begin{lemma}
  \label{lem:J_F_p_bound_better}
  Suppose $K\step \Lparam^{\frac12} \leq  \frac{1}{4}$, then for $1\leq j \leq K$, $\D_1 \Fo_j(q, p)$ and $\D_2 \Fo_j(q, p)$ satisfies 
  \begin{align}
    \label{eq:D2_Foj}
    &\D_1 \Fo_j(q,p) = \Ind_\dims-\frac{j\step^2}{2} \hessf_q-\step^2\sum_{\ell=1}^{j-1}(j-\ell) \hessf_{\Fo_\ell(q,p)}\D_1\Fo_\ell(q,p), \\
    &\D_2 \Fo_j(q, p) = j\step \Ind_\dims - \step^2 \sum_{\ell=1}^{j-1} (j-\ell) \hessf_{\Fo_\ell(q, p)} \D_2 \Fo_\ell(q, p).
  \end{align}
  For any $q,p\in\real^\dims$, we have:
  \begin{align}
    &\matsnorm{\D_1\Fo_j(q,p)-\Ind_\dims}{2}\leq j^2 \step^2 \Lparam,\\
    &\matsnorm{\D_2 \Fo_j(q, p) - j\step \Ind_\dims}{2} \leq j^3 \step^3 \Lparam. 
  \end{align}
  In particular, for $j=K$, we have
  \begin{align*}
    \matsnorm{\D_1 \Fo(q,p)-\Ind_\dims}{2} \leq \frac1{16},\quad 
    \matsnorm{\D_2 \Fo(q, p) - K\step \Ind_\dims}{2} \leq \frac{1}{16} K\step.
  \end{align*}
  Additionally, the inverse map $\Go$ is well-defined, and we have
  \begin{equation*}
      \matsnorm{\D_1 \Go(q,p)}{2}\leq \frac{17}{15K\step},\quad \matsnorm{\D_2\Go(q,p)-\frac{1}{K\step}\Ind_{\dims}}{2}\leq \frac{1}{15K\step}.
  \end{equation*}
\end{lemma}

Lemma~\ref{lem:Fj_q_p_bound} controls the Lipschitz constant of $\Fo_j$ as a function of two variables.
\begin{lemma}
  \label{lem:Fj_q_p_bound}
  Suppose $K\step \Lparam^{\frac12} \leq  \frac{1}{2}$, then for $1\leq j \leq K$ we have
  \begin{align*}
    \enorm{\Fo_j(q, p) - \Fo_j(\tilq, \tilp)} \leq 2 \enorm{q - \tilq} + 2 j \step \enorm{p - \tilp}.
  \end{align*}
\end{lemma}

Lemma~\ref{lem:trace_logdet} is a common inequality that bounds the log-determinant.

\begin{lemma}\label{lem:trace_logdet}
Let $A \in \real^{\dims \times \dims}$ be a matrix such that $\vecnorm{A}{2}\leq\frac12$, then
\begin{align*}
- \log\det(\Ind_\dims - A) - \trace(A) \leq 2\trace\parenth{AA\tp}.
\end{align*} 
\end{lemma}

\subsection{Proof of Lemma~\ref{lem:proposal_KL_div_bound}}
\label{sub:proof_of_lem:proposal_KL_div_bound}
In this subsection, we prove Lemma~\ref{lem:proposal_KL_div_bound} which bounds the KL divergence between two proposal distributions, which boils down to bounding the distance between a reference Gaussian and perturbed Gaussian as observed in~\cite{bou-rabee_mixing_2023}. Since different smoothness assumptions were used in~\cite{bou-rabee_mixing_2023}, we adapt the proof under our smoothness assumptions. Before that, we need to simplify the expression $\Ind_\dims - \D_2 \Go(\tilq_0, \Fo(q_0, p)) \D_2 \Fo(q_0, p)$ for any $\tilq_0, q_0, p \in \real^\dims$.

\begin{lemma}
  \label{lem:D2D2_expression}
  For any $\tilq_0, q_0, p \in \real^\dims$, we have 
  \begin{align}
    \label{eq:claim_form_D2D2}
    \Ind_\dims - \D_2 \Go(\tilq_0, \Fo(q_0, p)) \D_2 \Fo(q_0, p) = Q_0^{-1} \parenth{Q_1 + Q_2},  
  \end{align}
  where with the introduction of $\tilp \defn \Go(\tilq_0, \Fo(q_0, p))$, the three matrices have the form
  \begin{align*}
    Q_0 &= \Ind_\dims - \frac{\step}{K} \sum_{j=1}^{K-1} (K-j) \hessf_{\Fo_j(\tilq_0, \tilp)} \D_2 \Fo_j(\tilq_0, \tilp), \\
    Q_1 &= \frac{\step}{K} \sum_{j=1}^{K-1} (K-j) \parenth{ \hessf_{\Fo_j(q_0, p)}  - \hessf_{\Fo_j(\tilq_0, \tilp)} }\D_2 \Fo_j(q_0, p), \\
    Q_2 &= \frac{\step}{K} \sum_{j=1}^{K-1} (K-j)  \hessf_{\Fo_j(\tilq_0, \tilp)} \parenth{ \D_2 \Fo_j(q_0, p) - \D_2 \Fo_j(\tilq_0, \tilp)}.
  \end{align*}
\end{lemma}
The proof of this Lemma is deferred to the end of this subsection. We now prove Lemma~\ref{lem:proposal_KL_div_bound}.  

\begin{proof}[Proof of Lemma~\ref{lem:proposal_KL_div_bound}]
After the change of variable $p = \Go(q_0, y)$ in Eq.~\eqref{eq:KL_div_first_step}, we obtain
\begin{align}
  \label{eq:KL_div_second_step}
  &\quad \kldiv{\proposal_{q_0}}{\proposal_{\tilq_0}} \notag \\
  &= \int \varphi(p) \brackets{\frac12 \parenth{- \enorm{p}^2 + \enorm{\Go(\tilq_0, \Fo(q_0, p))}^2 } + \log\det \D_2 \Go(q_0, \Fo(q_0, p)) - \log\det \D_2 \Go (\tilq_0, \Fo(q_0, p))  } dp
\end{align}
Note that
\begin{align*}
  \frac12 \parenth{- \enorm{p}^2 + \enorm{\Go(\tilq_0, \Fo(q_0, p))}^2 } = - p \tp \parenth{p - \Go(\tilq_0, \Fo(q_0, p))} + \frac{1}{2} \enorm{p - \Go(\tilq_0, \Fo(q_0, p))}^2
\end{align*}
The second term is easier to bound, for $p = \Go(q_0, y)$, we have 
\begin{align}
  \label{eq:G_tilG_bound}
  \enorm{\Go(q_0, y) - \Go(\tilq_0, y)} \leq \sup_{q \in \real^\dims} \vecnorm{{\D_1} \Go(q, y)}{2} \enorm{q_0 - \tilq_0} \overset{Lem.~\ref{lem:J_F_p_bound_better}}{\leq} \frac{17}{15 K\step}\enorm{q_0 - \tilq_0}.
\end{align}
The first term requires Gaussian integration by parts. Isolating the second term, from Eq.~\eqref{eq:KL_div_second_step} we have
\begin{align}
  \label{eq:KL_div_third_step}
  &\quad \kldiv{\proposal_{q_0}}{\proposal_{\tilq_0}} - \frac{1}{2} \int \varphi(p)  \enorm{p - \Go(\tilq_0, \Fo(q_0, p))}^2 dp \notag \\
  &= \int \varphi(p) \brackets{ - p \tp (p - \Go(\tilq_0, \Fo(q_0, p))) + \log\det \D_2 \Go(q_0, \Fo(q_0, p)) - \log\det \D_2 \Go (\tilq_0, \Fo(q_0, p))  } dp \notag \\
  &\overset{(i)}{=} \int \varphi(p) \brackets{ - p \tp (p - \Go(\tilq_0, \Fo(q_0, p))) - \log\det \parenth{\D_2 \Go (\tilq_0, \Fo(q_0, p)) \D_2 \Fo(q_0, p)}  } dp \notag \\
  &\overset{(ii)}{=}  \int \varphi(p) \braces{ - \trace\brackets{\Ind_\dims - \D_2 \Go(\tilq_0, \Fo(q_0, p)) \D_2 \Fo(q_0, p)} - \log\det \parenth{\D_2 \Go (\tilq_0, \Fo(q_0, p)) \D_2 \Fo(q_0, p)}  } dp \notag \\
  &\overset{(iii)}{\leq } \int 2\varphi(p) \vecnorm{{\Ind_\dims - \D_2 \Go(\tilq_0, \Fo(q_0, p)) \D_2 \Fo(q_0, p)} }{F}^2  dp
\end{align}
where (i) follows from the identity~\eqref{eq:JG_vs_JF}, (ii) follows from Gaussian integration by parts and the boundary term is $0$. (iii) follows from Lemma~\ref{lem:trace_logdet}.

Using the notation in Lemma~\ref{lem:D2D2_expression}, we have
\begin{align*}
  \Ind_\dims - \D_2 \Go(\tilq_0, \Fo(q_0, p)) \D_2 \Fo(q_0, p) = Q_0^{-1} \parenth{Q_1 + Q_2}.
\end{align*}
According to Lemma~\ref{lem:J_F_p_bound_better}, we have
\begin{align*}
  \vecnorm{\Ind_\dims - Q_0}{2} &= \vecnorm{\frac{\step}{K} \sum_{j=1}^{K-1} (K-j) \hessf_{\Fo_j(q_0, \tilp)} \D_2 \Fo_j(q_0, \tilp)}{2} \\
  &\leq \vecnorm{\frac{\step}{K} \sum_{j=1}^{K-1} (K-j) \Lparam \frac{9}{8} K\step \Ind_\dims }{2} \\
  &\leq \frac{9}{16} K^2 \step^2 \Lparam. 
\end{align*}
Consequently, for $K^2 \step^2 \leq \frac{1}{4}$, we obtain $\vecnorm{Q_0^{-1}}{2} \leq 2$. 
Next we bound $Q_1$.
\begin{align*}
  \vecnorm{Q_1}{F} &\leq \frac{\step}{K} \sum_{j=1}^{K-1} (K-j) \vecnorm{ \hessf_{\Fo_j(q_0, p)}  - \hessf_{\Fo_j(\tilq_0, \tilp)} }{F} \vecnorm{\D_2 \Fo_j(q_0, p)}{2} \\
  &\overset{(i)}{\leq} \frac{\step}{K} \sum_{j=1}^{K-1} (K-j) \frac{9}{8} j \step \vecnorm{ \hessf_{\Fo_j(q_0, p)}  - \hessf_{\Fo_j(\tilq_0, \tilp)} }{F} \\
  &\overset{(ii)}{\leq} \frac{\step}{K} \sum_{j=1}^{K-1} (K-j) \frac{9}{8} j \step \SSC \Lparam^{\frac32} \enorm{\Fo_j(q_0, p) - \Fo_j(\tilq_0, \tilp)} \\
  &\overset{(iii)}{\leq} \frac{\step}{K} \sum_{j=1}^{K-1} (K-j) \frac{9}{8} j \step \SSC \Lparam^{\frac32} \parenth{2\enorm{q_0 - \tilq_0} + 2j\step \enorm{p - \tilp}} \\
  &\overset{(iv)}{\leq} 10 K^2 \step^2 \SSC \Lparam^{\frac32} \enorm{q_0 - \tilq_0}. 
\end{align*} 
(i) follows from Lemma~\ref{lem:J_F_p_bound_better}. (ii) follows from the $\SSC \Lparam^{\frac32}$-strongly Hessian Lipschitz condition. (iii) follows from Lemma~\ref{lem:Fj_q_p_bound}. (iv) follows from the fact that $\tilp = \Go(\tilq_0, \Fo(q_0, p))$ satisfies the implicit equation
\begin{align*}
  \Fo(q_0, p) &= \tilq_0 + K\step \tilp - \frac{K\step^2}{2} \gradf(\tilq_0) - \step^2 \sum_{j=1}^{K-1} (K-j) \gradf(\Fo_j(\tilq_0, \tilp)) \\
  \Fo(q_0, p) &= q_0 + K\step p - \frac{K\step^2}{2} \gradf(q_0) - \step^2 \sum_{j=1}^{K-1} (K-j) \gradf(\Fo_j(q_0, p)).
\end{align*}
Subtracting the above two equations leads to 
\begin{align*}
  \enorm{p - \tilp} \leq \frac{1}{K \step} \parenth{\enorm{q_0  - \tilq_0} + \frac{K\step^2}{2} \Lparam \enorm{q_0 - \tilq_0} + \step^2 \sum_{j=1}^{K-1} (K-j) \parenth{2\enorm{q_0 - \tilq_0} + 2 j \step \enorm{p - \tilp}} },
\end{align*}
and $\enorm{p - \tilp} \leq \frac{4 }{K \step} \enorm{q_0  - \tilq_0}$. 

Similar to what we did for $Q_1$, we bound $Q_2$ as follows
\begin{align*}
  \vecnorm{Q_2}{F} &\leq \frac{\step}{K} \sum_{j=1}^{K-1} (K-j) \vecnorm{\hessf_{\Fo_j(\tilq_0, \tilp)}}{2} \vecnorm{\D_2 \Fo_j(q_0, p) - \D_2 \Fo_j(\tilq_0, \tilp)}{F} \\
  &\leq \frac{\step}{K} \sum_{j=1}^{K-1} (K-j) \Lparam \vecnorm{\D_2 \Fo_j(q_0, p) - \D_2 \Fo_j(\tilq_0, \tilp)}{F} \\
  &\overset{(i)}{\leq} \frac{\step}{K} \sum_{j=1}^{K-1} (K-j) \Lparam 10 j^3 \step^3 \SSC\Lparam^{\frac32} \enorm{q_0 - \tilq_0} \\
  &\leq 10 K^2 \step^2 \SSC \Lparam^{\frac32} \enorm{q_0 - \tilq_0}.
\end{align*}
To see why (i) hold, we have 
\begin{align*}
  &\quad \vecnorm{\D_2 \Fo_j(q_0, p) - \D_2 \Fo_j(\tilq_0, \tilp)}{F} \\
  &\overset{\eqref{eq:D2_Foj}}{=} \vecnorm{ \step^2 \sum_{\ell=1}^{j-1} (j-\ell) \brackets{ \parenth{\hessf_{\Fo_\ell(q_0, p)} - \hessf_{\Fo_\ell(\tilq_0, \tilp)} } \D_2 \Fo_\ell(q, p) + \hessf_{\Fo_\ell(\tilq_0, \tilp)} \parenth{\D_2\Fo_\ell(q_0, p) - \D_2\Fo_\ell(\tilq_0, \tilp)}  } }{F} \\
  &\leq \step^2 \sum_{\ell=1}^{j-1}(j-\ell) \frac{9}{8} \ell \step 10 \SSC\Lparam^{\frac32} \enorm{q_0 - \tilq_0} + \step^2 \Lparam \sum_{\ell=1}^{j-1} (j-\ell)  \vecnorm{\D_2 \Fo_\ell(q_0, p) - \D_2 \Fo_\ell(\tilq_0, \tilp)}{F},
\end{align*}
and by induction, we obtain
\begin{align*}
  \vecnorm{\D_2 \Fo_j(q_0, p) - \D_2 \Fo_j(\tilq_0, \tilp)}{F} \leq 10 j^3 \step^3 \SSC\Lparam^{\frac32} \enorm{q_0 - \tilq_0}. 
\end{align*}

Combining the bounds on $Q_0$, $Q_1$ and $Q_2$ into Eq.~\eqref{eq:claim_form_D2D2}, we obtain
\begin{align*}
  \vecnorm{\Ind_\dims - \D_2 \Go(\tilq_0, \Fo(q_0, p)) \D_2 \Fo(q_0, p)}{F} \leq 40 K^2 \step^2 \SSC \Lparam^{\frac32} \enorm{q_0 - \tilq_0}.
\end{align*}
Together with Eq.~\eqref{eq:G_tilG_bound}, we conclude that
\begin{align*}
  \kldiv{\proposal_{q_0}}{\proposal_{\tilq_0}} &\leq \frac12\cdot\frac{17^2}{ 15^2} \frac{\enorm{q_0 - \tilq_0}^2}{K^2\step^2} + 40^2 K^4 \step^4 \SSC^2 \Lparam^3 \enorm{q_0 - \tilq_0}^2  \\
  &\leq \frac{1}{64} +  K^6 \step^6 \SSC^2 \Lparam^3,
\end{align*}
for $\enorm{q_0 - \tilq_0} \leq \frac1{64} K\step$.
\end{proof}

\begin{proof}[Proof of Lemma~\ref{lem:D2D2_expression}]
  $\Go(\tilq_0, \Fo(q_0, p))$ satisfies the implicit equation
  \begin{align*}
    \Fo(q_0, p) &= \tilq_0 + K\step \Go(\tilq_0, \Fo(q_0, p)) - \frac{K\step^2}{2} \gradf(\tilq_0) - \step^2 \sum_{j=1}^{K-1} (K-j) \gradf(\Fo_j(\tilq_0, \Go(\tilq_0, \Fo(q_0, p)))) \\
    \Fo(q_0, p) &= q_0 + K\step p - \frac{K\step^2}{2} \gradf(q_0) - \step^2 \sum_{j=1}^{K-1} (K-j) \gradf(\Fo_j(q_0, p)).
  \end{align*}
  Taking the difference, we obtain
  \begin{align*}
    &\quad p - \Go(\tilq_0, \Fo(q_0, p)) \\
    &= \frac{1}{K\step} \brackets{\frac{K\step^2}{2} \parenth{\gradf(q_0) - \gradf(\tilq_0)} + \step^2\sum_{j=1}^{K-1}(K-j) \parenth{\gradf(\Fo_j(q_0, p)) - \gradf(\Fo_j(\tilq_0, \Go(\tilq_0, \Fo(q_0, p)) ))}  }.
  \end{align*}
  Take derivative with respect to $p$, we obtain
  \begin{align*}
    &\quad \Ind_\dims - \D_2\Go(\tilq_0, \Fo(q_0, p)) \D_2 \Fo(q_0, p) \\
    &= \frac{\step}{K} \brackets{ \sum_{j=1}^{K-1}(K-j) \parenth{\hessf_{\Fo_j(q_0, p)} \D_2 \Fo_j(q_0, p)  - \hessf_{\Fo_j(\tilq_0, \tilp) } \D_2 \Fo_j(\tilq_0, \tilp) \D_2 \Go(\tilq_0, \Fo(q_0, p)) \D_2 \Fo(q_0, p)  }  },
  \end{align*}
  where $\tilp = \Go(\tilq_0, \Fo(q_0, p))$. Rearranging the terms that depend on $ \D_2 \Go(\tilq_0, \Fo(q_0, p)) \D_2 \Fo(q_0, p)$, we obtain the claimed result
  \begin{align*}
    &\quad \parenth{\Ind_\dims - \frac{\step}{K} \sum_{j=1}^{K-1} (K-j) \hessf_{\Fo_j(q_0, \tilp)} \D_2 \Fo_j(q_0, \tilp)  } \parenth{\Ind_\dims - \D_2 \Go(\tilq_0, \Fo(q_0, p)) \D_2 \Fo(q_0, p)} \notag \\
    &= \frac{\step}{K} \sum_{j=1}^{K-1} (K-j) \brackets{ \hessf_{\Fo_j(q_0, p)} \D_2 \Fo_j(q_0, p) - \hessf_{\Fo_j(\tilq_0, \tilp)} \D_2 \Fo_j(\tilq_0, \tilp)   } \notag  \\
    &= \frac{\step}{K} \sum_{j=1}^{K-1} (K-j) \brackets{ \parenth{ \hessf_{\Fo_j(q_0, p)}  - \hessf_{\Fo_j(\tilq_0, \tilp)} }\D_2 \Fo_j(q_0, p) + \hessf_{\Fo_j(\tilq_0, \tilp)} \parenth{ \D_2 \Fo_j(q_0, p) - \D_2 \Fo_j(\tilq_0, \tilp)}   }.
  \end{align*}
\end{proof}

\section{Acceptance rate}
\label{sec:acceptance}
In this section, we lower bound the HMC acceptance rate $\Accept_{K, \step}(q_0, p_0)$ in Eq.~\eqref{eq:HMC_acceptance} and derive the conditions needed on the choices of $K$ and $\step$.

\subsection{Notations and proof strategy}
We use $\target \times \Normal(0, \Ind_\dims)$ to denote the product measure of the target measure $\target$ and an independent standard Gaussian.
We introduce a single leapfrog step of length $\step$ of the HMC dynamics starting from $q_0, p_0$
\begin{align}
  \label{eq:def_q_0_step}
  q_{0, \step} &\defn q_0 + \step p_0 - \frac{\step^2}{2} \gradf(q_0) \notag \\
  p_{0, \step} &\defn p_0 - \frac{\step}{2} \gradf(q_0) - \frac{\step}{2} \gradf(q_{0, \step}).
\end{align}
Note that $(q_{0,0}, p_{0, 0}) = (q_0, p_0)$ and $(q_{0, \step}, p_{0, \step}) = (q_1, p_1)$. To simplify the notation, when the context is clear, we allow ourselves to overload the notation and to use $(q_\step, p_\step)$ as a shorthand for $(q_{0,\step}, p_{0, \step})$. 
\paragraph{Proof strategy.} First, observe that we have
\begin{align*}
  -f(q_K) - \frac{1}{2} \vecnorm{p_K}{2}^2 + f(q_0) + \frac{1}{2} \vecnorm{p_0}{2}^2 = \sum_{k=0}^{K-1} \brackets{-f(q_{k+1}) - \frac{1}{2}\vecnorm{p_{k+1}}{2}^2 + f(q_k) + \frac{1}{2} \vecnorm{p_k}{2}^2}.
\end{align*}
Hence, it is sufficient to bound each following term one by one inductively,
\begin{align*}
  f(q_{k+1}) - f(q_k) + \frac{1}{2}\vecnorm{p_{k+1}}{2}^2 - \frac{1}{2} \vecnorm{p_k}{2}^2.
\end{align*}
\begin{itemize}
  \item We first bound each of them by first assuming $(q_k, p_k)$ has approximately the same law as $\target \times \Normal(0, \Ind_\dims)$.
  \item We connect all the bounds at different $k$'s by showing inductively that as long as $(q_{k-1}, p_{k-1})$ has approximately the same law as $\target \times \Normal(0, \Ind_\dims)$, $(q_k, p_k)$ will have approximately the same law as well.
\end{itemize}

\begin{lemma}
  \label{lem:single_leapfrog_acceptance_lower_bound}
  Under Assumption~\ref{ass:assumption_main} on the target density, let integer $\ell \geq 1$ and let $\dl = \dims + 2(\ell-1)$, there exists a universal constant $c > 0$ such that a single leapfrog step with step-size $\step$ satisfies 
  \begin{align*}
    &\quad \brackets{\Exs_{(q_0, p_0) \sim \mu \times \Normal(0, \Ind_\dims)} \parenth{- f(q_\step) - \frac{1}{2} \enorm{p_\step}^2 + f(q_0) + \frac{1}{2} \enorm{p_0}^2}^{\ell}}^{\frac{1}{\ell}} \\
    &\leq  c (\SSC+1) \parenth{  \step^3 \ell^{\frac32} \Lparam^{\frac32} \dl^{\frac12} + \step^5 \ell^{\frac12} \Lparam^{\frac{5}{2}} \dl + \step^7 \Lparam^{\frac72} \dl^{\frac32}}.
  \end{align*}
\end{lemma}

\subsection{Intermediate concentration results}
\label{sub:intermediate_concentration_results}
We state the intermediate concentration results needed to prove Lemma~\ref{lem:single_leapfrog_acceptance_lower_bound}. The proofs are postponed to Appendix~\ref{app:add_proofs_acceptance}. 
To simplify notation, with $\trH$ and $\ell$, define the constant
\begin{align}
  \label{eq:def_Cl}
  \Cl \defn \trH + 2 (\ell-1) \Lparam.
\end{align}
With $\dims$ and $\ell$, define the constant
\begin{align}
  \label{eq:def_dl}
  \dl \defn \parenth{\dims + 2\ell-2} .
\end{align}
\begin{lemma}
  \label{lem:grad_norm_bound}
  Let integer $\ell \geq 1$. Suppose $e^{-f}$ is subexponential-tailed, $f$ is $\Lparam$-smooth and $\sup_{x \in \real^\dims} \trace(\hessf_x) \leq \trH$, let $\Cl$ be defined in Eq.~\eqref{eq:def_Cl}, then
  \begin{align*}
    \brackets{\Exs_{q \sim e^{-f}} \enorm{\gradf(q)}^{2\ell}}^{\frac{1}{\ell}} \leq \Cl \leq \Lparam \dl.
  \end{align*}
\end{lemma}
\begin{lemma}
  \label{lem:pHp_bound}
  Let integer $\ell \geq 1$ and $x\in \real^\dims$. Suppose $f$ is $L$-smooth, then
  \begin{align*}
    \brackets{\Exs_{p \sim \Normal(0, \Ind_\dims)} \parenth{p\tp \parenth{\hessf_{x}}^2 p}^{\ell}}^{\frac{1}{\ell}} \leq \parenth{\sup_{x \in \real^\dims} \trace( (\hessf_x)^2) + 2(\ell-1) L^2} \leq \Lparam^2 \dl.
  \end{align*}
\end{lemma}


\begin{lemma}
  \label{lem:gradHp}
  Suppose $e^{-f}$ is subexponential-tailed, $f$ is $\Lparam$-smooth and $\sup_{x \in \real^\dims} \trace(\hessf_x) \leq \trH$, let $\Cl$ be defined in Eq.~\eqref{eq:def_Cl}, then
  \begin{align*}
    \brackets{\Exs_{(q, p) \sim {e^{-f} \times \Normal(0, \Ind_\dims)}} \parenth{\gradf(q) \tp \hessf_{q} p}^\ell }^{\frac{1}{\ell}} \leq \ell^{\frac12} L\Cl^{\frac12} \leq \ell^{\frac12} \Lparam^{\frac32} \dl^{\frac12}
  \end{align*}
\end{lemma}

In the course of analyzing the acceptance rate of HMC, we need several concentrations bounds that involve the third order derivatives. The proof of these bounds are postponed to Appendix~\ref{sub:add_proofs_acceptance_third_order_bounds}.
\begin{lemma}
  \label{lem:jerk_ppp_bound}
  For any $x\in \real^\dims$, there exists a universal constant $c > 0$ such that 
  \begin{align*}
    \brackets{\Exs_{p \sim \Normal(0, \Ind_\dims)} \parenth{\jerkf_{x}[p, p, p]}^{\ell}}^{\frac{1}{\ell}} \leq c \parenth{\ell^{\frac32} \vecnorm{\jerkf_x}{\braces{123}} + \ell^{\frac12} \dims^{\frac12} \vecnorm{\jerkf_x}{\braces{12}\braces{3} } }.
  \end{align*}
  The bound is $O\parenth{\SSC \ell^{\frac32}  \Lparam^{\frac32} \dims^{\frac12}}$ if $f$ is $\SSC \Lparam^{\frac32}$-strongly Hessian Lipschitz.
\end{lemma}

\begin{lemma}
  \label{lem:jerk_pp_norm_bound}
  For any $x\in \real^\dims$, there exists a universal constant $c > 0$ such that 
  \begin{align*}
    \brackets{\Exs_{p \sim \Normal(0, \Ind_\dims)} \enorm{\jerkf_{x}[p, p, \cdot]}^{2\ell}}^{\frac{1}{\ell}} \leq c \parenth{\ell^2\vecnorm{\jerkf_x}{\braces{123}}^2 +  \ell^2 \dims \vecnorm{\jerkf_x}{\braces{12}\braces{3} }^2 }.
  \end{align*}
  The bound is $O\parenth{\SSC^2  \ell^2 \Lparam^3 \dims }$ if $f$ is $\SSC \Lparam^{\frac32}$-strongly Hessian Lipschitz.
\end{lemma}

\begin{lemma}
  \label{lem:pHp_diff_bound}
  Given $(\cq_0, \cp_0)$, let $(\cq_t, \cp_t)$ be the continuous Hamiltonian dynamics at time $t$ according to Eq.~\eqref{eq:continuous_HMC_form}, then there exists a universal constant $c>0$ such that 
  \begin{align*}
    &\quad \brackets{\Exs_{(\cq_0, \cp_0) \sim \target \times \Normal(0, \Ind_\dims)} \parenth{\cp_t\tp \hessf_{\cq_t} \cp_t - \cp_0 \tp \hessf_{\cq_0} \cp_0}^{\ell} }^{\frac{1}{\ell}} \\
    &\leq c t \parenth{\ell^{\frac12} \Lparam \Cl^{\frac12} +   \ell^{\frac32} \vecnorm{\jerkf_\cdot}{\braces{123}} + \ell^{\frac12} \dims^{\frac12} \vecnorm{\jerkf_\cdot}{\braces{12}\braces{3} } }.
  \end{align*}
  The bound is $O\parenth{t \parenth{\SSC + 1} \ell^{\frac32} \Lparam^{\frac32} \dl^{\frac12}}$ if $f$ is $\SSC \Lparam^{\frac32}$-strongly Hessian Lipschitz.
\end{lemma}

\begin{lemma}
  \label{lem:Hp_diff_bound}
  Given $(\cq_0, \cp_0)$, let $(\cq_t, \cp_t)$ be the continuous Hamiltonian dynamics at time $t$ according to Eq.~\eqref{eq:continuous_HMC_form}, then there exists a universal constant $c>0$ such that 
  \begin{align*}
    \brackets{\Exs \enorm{\hessf_{\cq_t} \cp_t - \hessf_{\cq_0} \cp_0}^{2\ell}}^{\frac{1}{2\ell}} &\leq c t \parenth{\vecnorm{\jerkf_\cdot}{\braces{123}} + \dims^{\frac12} \vecnorm{\jerkf_\cdot}{\braces{12}\braces{3}} + \Lparam\Cl^{\frac12} }
  \end{align*}
  The bound is $O\parenth{t \parenth{\SSC + 1} \ell^{\frac12} \Lparam^{\frac32} \dl^{\frac12}}$ if $f$ is $\SSC \Lparam^{\frac32}$-strongly Hessian Lipschitz.
\end{lemma}

\begin{lemma}
  \label{lem:qp_discrete_continuous_diff}
  Given $(\cq_0, \cp_0)$, let $(\cq_t, \cp_t)$ be the continuous Hamiltonian dynamics at time $t > 0$ according to Eq.~\eqref{eq:continuous_HMC_form} and $(q_t, p_t)$ be the defined in Eq.~\eqref{eq:def_q_0_step} with the same initialization, we have
  \begin{align*}
    \matsnorm{\cq_t - q_t}{2\ell} & \leq  t^3 \Lparam \dl^{\frac12} \\
    \matsnorm{\cp_t - p_t}{2\ell} &\leq  ct^3 \parenth{\vecnorm{\jerkf_\cdot}{\braces{123}} + \dims^{\frac12} \vecnorm{\jerkf_\cdot}{\braces{12}\braces{3}} + \Lparam^{\frac32}\dl^{\frac12}}.
  \end{align*}
  The second bound is $O\parenth{t^3 \parenth{\SSC + 1} \Lparam^{\frac32} \dl^{\frac12}}$ if $f$ is $\SSC \Lparam^{\frac32}$-strongly Hessian Lipschitz.
\end{lemma}

\subsection{Proof of Lemma~\ref{lem:single_leapfrog_acceptance_lower_bound}}
Lemma~\ref{lem:single_leapfrog_acceptance_lower_bound} controls the acceptance rate, which is the key technical novelty of this paper. It is proved via a careful treatment of the terms in the acceptance rate using the high dimensional concentration results in the above subsection. 

The Hamiltonian difference for a single leapfrog step is
\begin{align*}
  - f(q_\step) + f(q_0) - \frac12 \enorm{p_\step}^2 + \frac12 \enorm{p_0}^2.
\end{align*}
Let integer $\ell \geq 1$, we are interested in bounding $\matsnorm{- f(q_\step) + f(q_0) - \frac12 \enorm{p_\step}^2 + \frac12 \enorm{p_0}^2}{\ell}$.
Since the continuous HMC dynamics conserves the Hamiltonian, we have
\begin{align*}
  &\quad - f(q_\step) + f(q_0) - \frac12 \enorm{p_\step}^2 + \frac12 \enorm{p_0}^2 \\
  &= - f(q_\step) + f(\cq_\step) - \frac12 \enorm{p_\step}^2 + \frac12 \enorm{\cp_\step}^2\\
  &= \int_0^1 \gradf(\omega \cq_\step + (1-\omega)q_\step) \tp (\cq_\step - q_\step) d\omega +  \frac12 \parenth{\cp_\step + p_\step}\tp \parenth{\cp_\step - p_\step}
\end{align*}
According to Lemma~\ref{lem:qp_discrete_continuous_diff}, $\cq_\step - q_\step$ is approximately $\step^3\Lparam \dl^{\frac12}$ close and so as $\cp_\step - p_\step$. Hence, we have
\begin{align}
  \label{eq:single_leapfrog_cont_version_bigO}
  \matsnorm{- f(q_\step) + f(q_0) - \frac12 \enorm{p_\step}^2 + \frac12 \enorm{p_0}^2}{\ell} = \matsnorm{\gradf(\cq_\step) \tp (\cq_\step - q_\step) +  \cp_\step \tp \parenth{\cp_\step - p_\step}}{\ell} + O(\step^6 \Lparam^3 \dl).
\end{align}
It remains to Bound
\begin{align}
  \label{eq:single_leapfrog_cont_version}
  \gradf(\cq_\step) \tp (\cq_\step - q_\step) +  \cp_\step \tp \parenth{\cp_\step - p_\step}.
\end{align}
Recall from Eq.~\eqref{eq:continuous_HMC_form} and Eq.~\eqref{eq:def_q_0_step} that
\begin{align*}
  \cq_\step &= q_0 + \step p_0 - \int_0^\step \int_0^t \gradf(\cq_\tau) d\tau dt \\
  q_\step &= q_0 + \step p_0 - \frac{\step^2}{2} \gradf(q_0) \\
  \cp_\step &= p_0 - \int_0^\step \gradf(\cq_t) dt \\
  p_\step &= p_0 - \frac{\step}{2} \gradf(q_0) - \frac{\step}{2} \gradf(q_\step).
\end{align*}
We have that
\begin{align*}
  \cq_\step - q_\step &= - \int_0^\step \int_0^t (\gradf(\cq_s) - \gradf(\cq_0)) ds dt\\
  &= - \int_0^\step \int_0^t \int_0^s \hessf_{\cq_\tau} \cp_\tau d\tau ds dt \\
  \cp_\step - p_\step &= - \int_0^\step \gradf(\cq_t) - \frac12 \gradf(\cq_0) - \frac12 \gradf(\cq_\step) dt  -  \frac{\step}{2} \parenth{\gradf(\cq_\step) - \gradf(q_\step) } \\
  &= - \int_0^\step \parenth{\frac12 \step - t} \hessf_{\cq_t} \cp_t dt - \frac{\step}{2} \parenth{\gradf(\cq_\step) - \gradf(q_\step) }\\
  \cp_\step \tp (\cp_\step - p_\step) 
  &= \cp_\step \tp \parenth{- \int_0^\step \parenth{\frac12 \step - t} \hessf_{\cq_t} \cp_t dt - \frac{\step}{2} \hessf_{q_{\text{mid}}} \parenth{\cq_\step - q_\step}}
\end{align*}
where there exists $\omega \in [0,1]$ such that $q_{\text{mid}} = \omega \cq_\step + (1-\omega) q_\step$. 

In the first part of Eq.~\eqref{eq:single_leapfrog_cont_version}, we have
\begin{align*}
  &\quad \gradf(\cq_\step) \tp (\cq_\step - q_\step) \\
  &=  -\int_0^\step \int_0^t \int_0^s \gradf(\cq_\step) \tp \hessf_{\cq_\tau} \cp_\tau d\tau ds dt \\
  &= - \int_0^\step \int_0^t \int_0^s \parenth{\gradf(\cq_\step) - \gradf(\cq_\tau)} \tp \hessf_{\cq_\tau} \cp_\tau d\tau ds dt - \underbrace{\int_0^\step \int_0^t \int_0^s \gradf(\cq_\tau) \tp \hessf_{\cq_\tau} \cp_\tau d\tau ds dt}_{A_1 \lesssim \step^3 \ell^{\frac12} \Lparam^{\frac32} \dl^{\frac12} \text{ from Lem.~\ref{lem:gradHp}}} \\
  &= - \int_0^\step \int_0^t \int_0^s \parenth{\int_{\tau}^\step \hessf_{\cq_\iota} \cp_\iota d\iota} \tp \hessf_{\cq_\tau} \cp_\tau d\tau ds dt + A_1 \\
  &= - \int_0^\step \int_0^t \int_0^s \parenth{\int_{\tau}^\step \hessf_{\cq_\iota} \cp_\iota d\iota} \tp \parenth{\hessf_{\cq_\tau} \cp_\tau - \hessf_{\cq_0} \cp_0} d\tau ds dt  \\
  &\quad + \int_0^\step \int_0^t \int_0^s \parenth{\int_{\tau}^\step \parenth{\hessf_{\cq_\iota} \cp_\iota - \hessf_{\cq_0} \cp_0} d\iota} \tp  \hessf_{\cq_0} \cp_0 d\tau ds dt \\
  &\quad - \frac{3\step^4}{24} \enorm{\hessf_{\cq_0} \cp_0}^2 + A_1
\end{align*}
Applying the bound $\matsnorm{\hessf_{\cq_\tau} \cp_\tau - \hessf_{\cq_0} \cp_0}{\ell} \lesssim \tau \parenth{\SSC + 1} \ell^{\frac12} \Lparam^{\frac32} \dl^{\frac12}$ in Lemma~\ref{lem:Hp_diff_bound} and the bound $\matsnorm{\hessf_{\cq_0} \cp_0}{2\ell} \lesssim \Lparam \dl^{\frac12}$ in Lemma~\ref{lem:pHp_bound}, we obtain 
\begin{align}
  \label{eq:first_part_in_accept_reject_lem8}
  \matsnorm{\gradf(\cq_\step) \tp (\cq_\step - q_\step) + \frac{3\eta^4}{24} \enorm{\hessf_{\cq_0} \cp_0}^2}{\ell} =  O\parenth{\step^3 \ell^{\frac12} \Lparam^{\frac32} \dl^{\frac12} + \parenth{\SSC + 1} \step^5 \ell^{\frac12} \Lparam^{\frac{5}{2}} \dl }.
\end{align}

In the term $\cp_\step \tp (\cp_\step - p_\step)$, the first half is
\begin{align*}
  &\cp_\step \tp \int_0^\step \parenth{t- \frac12 \step} \hessf_{\cq_t} \cp_t dt \\
  &= \int_0^\step \parenth{t - \frac12 \step} \parenth{\cp_\step - \cp_t}  \tp \hessf_{\cq_t} \cp_t dt +  \int_0^\step \parenth{t- \frac12 \step} \cp_t \tp \hessf_{\cq_t} \cp_t dt \\
  &= \int_0^\step \parenth{t - \frac12 \step} \parenth{\cp_\step - \cp_t}  \tp \hessf_{\cq_t} \cp_t dt - \underbrace{\int_0^\step \parenth{t - \frac12 \step} \brackets{\cp_t \tp \hessf_{\cq_t} \cp_t - \cp_0 \tp \hessf_{\cq_0} \cp_0}}_{A_2 \lesssim \parenth{\SSC + 1} \step^3 \ell^{\frac32} \Lparam^{\frac32} \dl^{\frac12} \text{ from Lem.~\ref{lem:pHp_diff_bound}}} dt
\end{align*}
The rest of it
\begin{align*}
  &\quad \int_0^\step \parenth{t - \frac12 \step} \parenth{\cp_\step - \cp_t}  \tp \hessf_{\cq_t} \cp_t dt \\
  &=\int_0^\step \parenth{t - \frac12 \step}  \int_t^\step -\gradf(\cq_s) \tp ds  \hessf_{\cq_t} \cp_t dt  \\
  &=\int_0^\step \parenth{t - \frac12 \step}  \int_t^\step \parenth{\gradf(\cq_t) -\gradf(\cq_s)} \tp ds  \hessf_{\cq_t} \cp_t dt - \underbrace{\int_0^\step \parenth{t - \frac12 \step}  (\step - t) \gradf(\cq_t) \tp \hessf_{\cq_t} \cp_t dt}_{A_3 \lesssim \step^3 \ell^{\frac12} \Lparam^{\frac32} \dl^{\frac12} \text{ from Lem.~\ref{lem:gradHp}}} \\
  &= \int_0^\step \parenth{t - \frac12 \step}  \parenth{\int_t^\step \int_s^t \hessf_{\cq_\tau} \cp_\tau d\tau ds} \tp \hessf_{\cq_t} \cp_t dt + A_3
\end{align*}
Applying the bound $\matsnorm{\hessf_{\cq_\tau} \cp_\tau - \hessf_{\cq_0} \cp_0}{\ell} \lesssim \tau \parenth{\SSC + 1} \ell^{\frac12} \Lparam^{\frac32} \dl^{\frac12}$ in Lemma~\ref{lem:Hp_diff_bound} and the bound $\matsnorm{\hessf_{\cq_0} \cp_0}{2\ell} \lesssim \Lparam \dl^{\frac12}$ in Lemma~\ref{lem:pHp_bound}, we obtain 
\begin{align}
  \label{eq:second_part_first_half_in_accept_reject_lem8}
  \matsnorm{\cp_\step \tp \int_0^\step \parenth{t- \frac12 \step} \hessf_{\cq_t} \cp_t dt - \frac{\eta^4}{24} \enorm{\hessf_{\cq_0} \cp_0}^2}{\ell} =  O\parenth{\parenth{\SSC + 1}\parenth{ \step^3 \ell^{\frac32} \Lparam^{\frac32} \dl^{\frac12} + \step^5 \ell^{\frac12} \Lparam^{\frac{5}{2}} \dl} }.
\end{align}

In the term $\cp_\step \tp (\cp_\step - p_\step)$, the second half is
\begin{equation*}
  - \frac{\step}{2} \cp_\step\tp \hessf_{q_{\text{mid}}} (\cq_\step - q_\step) 
  =  \frac{\step}{2} \cp_\step\tp \hessf_{q_{\text{mid}}}  \int_0^\step \int_0^t \int_0^s \hessf_{\cq_\tau} \cp_\tau d\tau ds dt.    
\end{equation*}

\begin{align*}
  \enorm{\hessf_{q_{\text{mid}}} \cp_\step - \hessf_{\cq_\step} \cp_\step} &\leq \enorm{\hessf_{q_{\text{mid}}} - \hessf_{\cq_\step}} \enorm{\cp_\step} \\
  &\leq \vecnorm{\jerkf_\cdot}{\braces{12}\braces{3}} \enorm{q_{\text{mid}} - \cq_\step} \enorm{\cp_\step}\\
  &\leq \vecnorm{\jerkf_\cdot}{\braces{12}\braces{3}} \enorm{\cq_\step - q_\step} \enorm{\cp_\step}.
\end{align*}
Applying the bound $\matsnorm{\hessf_{\cq_\tau} \cp_\tau - \hessf_{\cq_0} \cp_0}{\ell} \lesssim \tau \parenth{\SSC + 1} \ell^{\frac12} \Lparam^{\frac32} \dl^{\frac12}$ in Lemma~\ref{lem:Hp_diff_bound}, the bound $\matsnorm{\hessf_{\cq_0} \cp_0}{2\ell} \lesssim \Lparam \dl^{\frac12}$ in Lemma~\ref{lem:pHp_bound} and the bound $\matsnorm{\cq_\step - q_\step}{\ell} \lesssim \step^3 \Lparam \dl^{\frac12}$ in Lemma~\ref{lem:qp_discrete_continuous_diff}, we obtain 
\begin{align}
  \label{eq:second_part_second_half_in_accept_reject_lem8}
  \matsnorm{- \frac{\step}{2} \cp_\step\tp \hessf_{q_{\text{mid}}} (\cq_\step - q_\step)- \frac{\eta^4}{12} \enorm{\hessf_{\cq_0} \cp_0}^2}{\ell} =  O\parenth{\parenth{\SSC + 1} \parenth{ \step^5 \ell^{\frac12} \Lparam^{\frac{5}{2}} \dl + \step^7 \Lparam^{\frac72} \dl^{\frac32}}}.
\end{align}
Overall, all the terms with $\enorm{\hessf_{\cq_0} \cp_0}^2$ gets perfectly cancelled, combining Eq.~\eqref{eq:single_leapfrog_cont_version_bigO},~\eqref{eq:first_part_in_accept_reject_lem8},~\eqref{eq:second_part_first_half_in_accept_reject_lem8} and~\eqref{eq:second_part_second_half_in_accept_reject_lem8}, we obtain that the Hamiltonian difference for a single leapfrog step is
\begin{align*}
  \matsnorm{- f(q_\step) + f(q_0) - \frac12 \enorm{p_\step}^2 + \frac12 \enorm{p_0}^2}{\ell} = 
  O\parenth{ \parenth{\SSC + 1} \parenth{ \step^3 \ell^{\frac32} \Lparam^{\frac32} \dl^{\frac12} + \step^5 \ell^{\frac12} \Lparam^{\frac{5}{2}} \dl + \step^7 \Lparam^{\frac72} \dl^{\frac32}}}. 
\end{align*}


\subsection{Proof of Lemma~\ref{lem:acceptance_lower_bound}}
\label{sub:proof_of_lem:acceptance_lower_bound}
The proof of Lemma~\ref{lem:acceptance_lower_bound} relies on the fact that $(q_k, p_k)$ is approximately distributed as $\target \times \Normal(0, \Ind_\dims)$ for any $k$. We are going to show this by induction over $k$. 

Let $\alpha = \frac{\log(2)}{K}$ be a constant. Define the events $W_k$ as follows
\begin{align*}
  W_k \defn \bigcap_{j=0}^k \braces{(q_0, p_0) \in \real^\dims \times \real^\dims \mid \abss{\Ham(q_{j+1}, p_{j+1}) - \Ham(q_j, p_j)} \leq \alpha}
\end{align*}
For any test function $h$, $k\geq 1$ and $\ell \geq 2$ even integer. Conditioned on $W_{k-1}$, we have
\begin{align*}
  \Ham(q_k, p_k) - \Ham(q_0, p_0) \leq \log(2).  
\end{align*}
Then 
\begin{align}
  \label{eq:h_conditional_prob}
  &\quad \Exs_{(q_0, p_0) \sim \mu} \braces{h(q_k, p_k)^\ell \mathbf{1}_{(q_0, p_0) \in W_{k-1} }}  \notag \\
  &= \int \int \braces{h(q_k, p_k)^\ell \mathbf{1}_{(q_0, p_0) \in W_{k-1} }} \mu(q_0, p_0) dq_0 dp_0 \notag \\
  &= \int \int \braces{h(q_k, p_k)^\ell \mathbf{1}_{(q_0, p_0) \in W_{k-1} }} e^{-\Ham(q_k, p_k)} e^{-\Ham(q_0, p_0) + \Ham(q_k, p_k)} dq_0 dp_0 \notag\\
  &\overset{(i)}{\leq}  2 \int \int \braces{h(q_k, p_k)^\ell \mathbf{1}_{(q_0, p_0) \in W_{k-1} }} e^{-\Ham(q_k, p_k)} dq_0 dp_0 \notag \\
  &\leq 2 \int \int \braces{h(q_k, p_k)^\ell } e^{-\Ham(q_k, p_k)} dq_0 dp_0 \notag\\
  &\overset{(ii)}{=}  2 \int \int \braces{h(q_k, p_k)^\ell } e^{-\Ham(q_k, p_k)} dq_k dp_k \notag \\
  &= 2 \Exs_{(q, p) \sim \mu} \braces{h(q, p)^\ell}.
\end{align} 
(i) follows from the previous Hamiltonian difference conditioned on $W_{k-1}$. (ii) makes use of the change of variable $(q_0, p_0) \to (q_k,p_k)$ and the Jacobian is 1. In words, the above derivation shows that conditioned on the event $W_{k-1}$, we can upper-bound the expectation of any quantities derived from $(q_k, p_k)$ as if $(q_k, p_k)$ were distributed as $\mu$, up to a factor of 2. 

Next with the step-size choices in the lemma statement, we prove by induction that $\Prob(W_{K-1})$ is large. We have
\begin{align*}
    \Prob_{(q_0,p_0)\sim \mu}(W_k^c) &= \Exs_{(q_0, p_0) \sim \mu}\braces{\mathbf{1}_{(q_0,p_0)\notin W_{k}}}\\
    & = \Exs_{(q_0, p_0) \sim \mu}\braces{\mathbf{1}_{(q_0,p_0)\notin W_{k}} \mathbf{1}_{(q_0,p_0)\in W_{k-1}}} + \Exs_{(q_0, p_0) \sim \mu}\braces{\mathbf{1}_{(q_0,p_0)\notin W_{k-1}}}\\
    &=\Exs_{(q_0, p_0) \sim \mu}\braces{\mathbf{1}_{ \abss{\mathcal H(q_k, p_k) - \mathcal H(q_{k-1}, p_{k-1})} > \alpha} \mathbf{1}_{(q_0,p_0)\in W_{k-1}}} + \Prob_{(q_0,p_0)\sim \mu}(W_{k-1}^c)\\
    &\leq 2\Exs_{(q, p) \sim \mu}\braces{\mathbf{1}_{\abss{\mathcal H(\F_1(q, p)) - \mathcal H(q, p)} > \alpha }} + \Prob_{(q_0,p_0)\sim \mu}(W_{k-1}^c)\\
    &= 2\Prob_{(q, p) \sim \mu}\parenth{\braces{ \abss{\mathcal H(\F_1(q, p)) - \mathcal H(q, p)} > \alpha}} + \Prob_{(q_0,p_0)\sim \mu}(W_{k-1}^c).
\end{align*}
Here $\F_1$ is the forward map that runs a single leapfrog step of HMC with length $\step$, defined below Eq.~\eqref{eq:leapfrog_HMC_recursion}. Applying Markov's inequality, we obtain 
\begin{align*}
  &\Prob_{(q, p) \sim \mu}\parenth{\braces{ \abss{\mathcal H(\F_1(q, p)) - \mathcal H(q, p)} > \alpha}} \\
  \leq &\frac{{\Exs_{(q_0, p_0) \sim \mu \times \Normal(0, \Ind_\dims)} \parenth{- f(q_\step) - \frac{1}{2} \enorm{p_\step}^2 + f(q_0) + \frac{1}{2} \enorm{p_0}^2}^{\ell}}}{\alpha^{\ell}} \\
  \overset{(i)}{\leq} &\frac{c^\ell (\SSC+1)^\ell\parenth{  \step^3 \ell^{\frac32} \Lparam^{\frac32} \dl^{\frac12} + \step^5 \ell^{\frac12} \Lparam^{\frac{5}{2}} \dl + \step^7 \Lparam^{\frac72} \dl^{\frac32}}^\ell}{\alpha^{\ell}}.
\end{align*}
(i) follows from Lemma~\ref{lem:single_leapfrog_acceptance_lower_bound}. We control the error probability as long as
\begin{align*}
  \frac{c \parenth{\SSC + 1} \parenth{ \step^3 \ell^{\frac32} \Lparam^{\frac32} \dl^{\frac12} + \step^5 \ell^{\frac12} \Lparam^{\frac{5}{2}} \dl +  \step^7 \Lparam^{\frac72} \dl^{\frac32}} }{\alpha} \leq \parenth{\frac{\delta}{2K}}^{\frac{1}{\ell}}.
\end{align*}
It is sufficient to take $\ell = 2\ceils{\log(1/\delta)}$ and 
\begin{align*}
  K^{1+\frac1\ell} \parenth{\SSC + 1} \parenth{ \step^3 \ell^{\frac32} \Lparam^{\frac32} \dl^{\frac12} +  \step^5 \ell^{\frac12} \Lparam^{\frac{5}{2}} \dl + \step^7 \Lparam^{\frac72} \dl^{\frac32}} \leq \frac{1}{c'},
\end{align*}
for a constant $c'$.

\section{Examples of functions with small strongly Hessian Lipschitz coefficient}
\label{sec:examples}
In this section, we discuss several target density examples which have small strongly Hessian Lipschitz coefficient. We start by introducing ridge separable functions.

\paragraph{Ridge separable functions:}
we say a function $f: \real^\dims \to \real$ is \textit{ridge separable}, if there exist an integer $n\geq 1$ and univariate functions $\{u_i\}_{i=1}^n$, vectors $\{a_i\}$ in $\real^\dims$ such that
\begin{align}
    f(\theta) = \sum_{i=1}^n u_i(a_i^\top \theta).\label{eq:ridgedef}
\end{align}
As discussed in~\cite{mou2019high}, ridge separability naturally arises in posterior sampling in Bayesian statistical models such as Bayesian ridge regression, generalized linear models, etc. Lemma~\ref{lem:ridgeseparable} bounds various norms of the third order derivatives tensor of $f$. Its proof is provided in Appendix~\ref{app:separableexamples}.

\begin{lemma}[Third derivative tensor bound for ridge separable $f$]
  \label{lem:ridgeseparable}
Given a ridge separable function $f$ defined in Equation~\eqref{eq:ridgedef}, assume that $a_i$'s are of unit norm. Suppose we have second and third order derivative bounds for $u_i$, i.e. $\sup_{z\in \real}|u''_i(z)| \leq \rho_2$  and $\sup_{z\in \real}|u_i'''(z)| \leq \rho_3$. Then,  
\begin{gather*}
    \vecnorm{\hessf(\theta)}{2} \leq n \rho_2,\\
    \vecnorm{\jerkf(\theta)}{\braces{12}\braces{3}} \leq \vecnorm{\jerkf(\theta)}{\braces{123}} \leq n \rho_3.
\end{gather*}
\end{lemma}

Note that $n$ can be much smaller than $d$ in practice. When the second and third order derivatives of the activation are bounded as above, then $f$ is $\gamma L^{3/2}$-strongly Hessian Lipschitz for small $\gamma$, which means that using $K > 1$ in HMC is beneficial in reducing gradient complexity, as we discussed around Corollary~\ref{cor:best_HMC}.  
To see this argument quantitatively, we introduce a corollary that bounds the gradient complexity of HMC when $f$ is ridge separable and the corresponding density satisfies an isoperimetric inequality.

\begin{corollary}\label{cor:ridgesep}
Let $f$ be a ridge separable function defined in Equation~\eqref{eq:ridgedef}. Given that the target measure $\mu$ with density $\propto e^{-f}$ satisfies an isoperimetric inequality with constant $\psi_\mu$ and satisfies the derivative bounds in Lemma~\ref{lem:ridgeseparable}, then the gradient complexity of HMC with the optimal choices of free parameters as in Corollary~\ref{cor:best_HMC} is $\Ot\parenth{d^{\frac14}\brackets{n\rho_2 +(n\rho_3)^{\frac23}} {\psi}^{-2}_\mu}$.
\end{corollary}
\begin{proof}
    Note that from Lemma~\ref{lem:ridgeseparable}, $f$ is $L$ smooth and $\gamma L^{\frac{3}{2}}$-strongly Hessian Lipschitz with $L = n\rho_2$ and $\gamma = (n\rho_3)/(n\rho_2)^{3/2} = \rho_3/(n^{1/2}\rho_2^{3/2})$. The mixing time then follows directly from Corollary~\ref{cor:best_HMC}.
\end{proof}
Next, we present an example of a posterior distribution that arises in  Bayesian logistic regression where its log density is ridge separable.

\begin{lemma}[Logistic regression]\label{lem:logisticreg}
    Consider the logistic regression problem.
    We observe a set of data points $\braces{(x_i, y_i)}_{i=1}^n$ where $x_i \in \real^d$ and $y_i \in \braces{0,1}$. For simplicity, assume $\enorm{x_i} = 1, \ \forall i\in [n]$. The logistic loss is defined as
    \begin{align*}
        &\ell(z,1) \defn \log{\big(1 + \exp{\parenth{-z}}\big)},\\
        &\ell(z,0) \defn z +\log{\big(1 + \exp{\parenth{-z}}\big)}.
    \end{align*}
    The total loss on the dataset is given by
    \begin{align*}
        \mathcal L(\theta) \coloneqq 
        \sum_{i=1}^n \ell(\theta^\top x_i, y_i).
    \end{align*}
    Given $\braces{(x_i, y_i)}_{i=1}^n$, we would like to sample from the posterior distribution $\mu$ on $\theta$ with likelihood $\propto e^{-\mathcal L(\theta)}$ and prior $\Normal(0, \alpha^{-2} \Ind_\dims)$. That is, 
    \begin{align*}
      \mu(\theta) \propto e^{-\mathcal L(\theta) -\frac{\alpha^2}{2} \enorm{\theta}^2}.
    \end{align*} 
    Then from a warm start, HMC has gradient complexity $\tilde O(\dims^{\frac14} (n + \alpha^2)/\alpha^2)$.
\end{lemma}

Finally, we consider the problem of posterior sampling for two layer neural networks with fixed second-layer weights. This is another example where the log-density is ridge separable. 
\paragraph{Two-layer neural networks:} Let $f_{NN}$ denote the neural network function
\begin{align}
    f_{NN}(x) &\defn \sum_{j=1}^m w_j\sigma(\theta_j^\top x),\label{eq:twolayernetwork}
\end{align}
where $\sigma$ is an activation function $\real \to \real$. $\braces{w_j}_{j=1}^m$ are fixed weights in $\real$ and $\{\theta_j\}_{j=1}^m$ are the parameters where for any $j$, $\theta_j \in \mathbb R^{d'}$. Given a set of data points $\braces{(x_i, y_i)}_{i=1}^n$, the risk function under the squared loss is 
\begin{align}
  \label{eq:twolayerloss}
  \mathcal L(\theta) \defn \sum_{i=1}^n \parenth{y_i - f_{NN}(x_i)}^2.
\end{align}
Note that $f_{NN}$ is implicitly a function of $\theta$.
Lemma~\ref{lem:twolayernet} bounds the second and third derivatives of $\mathcal L$ for this two-layer network.
\begin{lemma}[Two layer neural  network]\label{lem:twolayernet}
    Consider the two-layer neural network model defined in Equation~\eqref{eq:twolayernetwork} with weights $\abss{w_j} \leq 1, \forall j \in [m]$ and centered data points having unit norm $\enorm{x_i} = 1, \forall i \in [n]$. Suppose that for the activation function $\sigma$, there exists a constant $c>0$ such that the derivatives up to third order are bounded as follows
    \begin{align*}
|\sigma(z)|,|\sigma'(z)|,|\sigma''(z)|,|\sigma'''(z)| \leq c, \forall z \in \real.
    \end{align*}
    Suppose that the labels are also bounded as $\abss{y_i} \leq mc, \forall i \in [n]$, 
Then the second and third order derivative tensors of $\mathcal L$ are bounded as follows

\begin{equation}\label{eq:derivativebounds}
\begin{aligned}
&\vecnorm{\nabla^2 \mathcal L(\theta)}{2} \lesssim mnc^2,\quad
\matsnorm{\nabla^3 \mathcal L(\theta)}{\braces{1}\braces{2}\braces{3}} \lesssim mnc^2\\
&\matsnorm{\nabla^3 \mathcal L(\theta)}{\braces{12}\braces{3}} \leq \matsnorm{\nabla^3 \mathcal L(\theta)}{\braces{123}} \lesssim m^2 nc^2.
\end{aligned}
\end{equation}

\end{lemma}

Note that $\theta$ has a dimension of $md'$, thus a na\"ive upper bound of
$\matsnorm{\nabla^3 \Law}{\braces{12}\braces{3}}$ by $\matsnorm{\nabla^3 \Law}{\braces{1}\braces{2}\braces{3}}$ would lose a factor of $\sqrt{md'}$. The bound in Eq.~\eqref{eq:derivativebounds} is tighter when $d'\gg m$.

Using the upper bounds in~\eqref{eq:derivativebounds} we get that $\mathcal{L}$ is $L$-smooth and $\gamma L^{3/2}$-strongly Hessian Lipschitz with $\SSC = O(1)$ assuming $c\geq 1$ and $m\leq n$. In this case, since $\SSC=O(1)$, according to our mixing time upper bound in Corollary~\ref{cor:best_HMC}, to sample from the distribution $e^{-\mathcal{L}(\theta)}$, HMC with $K > 1$ reduces the gradient complexity when compared to the choice $K=1$. We don't state an explicit mixing time bound because bounding the isoperimetric coefficient of posterior distribution $e^{-\mathcal{L}(\theta)}$ that arises from the two-layer network model is difficult in general.  It is worth noting that in the mean field approximation where $m$ tends to infinity, logarithmic Sobolev Inequalities were derived for the asymptotic distribution of a single weight in~\cite{chizat2022meanfield,zhang2022mean}. Satisfying a logarithmic Sobolev inequality implies satisfying an isoperimetric inequality. However, it is not clear whether any reasonable dimension-independent isoperimetric coefficient can be obtained via the same strategy for the joint distribution of weights in finite $m$ scenarios.


\bibliographystyle{alpha}
\bibliography{ref}

\newpage
\appendix
\input{appendix_hmc_ssc.tex}
\input{separableexamples}

\end{document}

%% file: macros.tex
\newcommand\numberthis{\addtocounter{equation}{1}\tag{\theequation}}

\newcommand{\target}{\mu}

\newcommand{\isop}{\psi}
\newcommand{\dist}{\mathfrak{D}}

\newcommand{\cq}{\mathfrak{q}}
\newcommand{\cp}{\mathfrak{p}}

\newcommand{\vq}{\mathbf{q}}

\newcommand{\tilq}{\tilde{q}}
\newcommand{\tilp}{\tilde{p}}

\newcommand{\proposal}{\mathcal{P}}
\newcommand{\transition}{\mathcal{T}}
\newcommand{\step}{\eta}
\newcommand{\tole}{\epsilon}
\newcommand{\taghmc}{\text{HMC}}
\newcommand{\tagmala}{\text{MALA}}
\newcommand{\tmix}{\tau_{\text{mix}}}
\newcommand{\tvdis}{{\rm{d}}_{{\rm{TV}}}}

\newcommand{\gradf}{\nabla f}
\newcommand{\hessf}{\nabla^2 f}
\newcommand{\jerkf}{\nabla^3 f}
\newcommand{\Ham}{\mathcal{H}}
\newcommand{\Accept}{\mathcal{A}}

\newcommand{\D}{\mathbf{D}}

\newcommand{\Lparam}{L}

\newcommand{\mparam}{m}
\newcommand{\condi}{\kappa}
\newcommand{\SSC}{\gamma}
\newcommand{\Cl}{\Upsilon_\ell}
\newcommand{\dl}{\dims_\ell}

\newcommand{\F}{F}
\newcommand{\bF}{\bar{F}}

\newcommand{\Fo}{\mathbb{F}}
\newcommand{\Go}{\mathbb{G}}

\newcommand{\trH}{\Upsilon}

\newcommand{\warm}{M}

\newcommand{\defn}{:=}

\newcommand{\polylog}{\text{polylog}}

\theoremstyle{plain}

\newtheorem{theorem}{Theorem}

\newtheorem{lemma}{Lemma}
\newtheorem{corollary}{Corollary}

\newtheorem{assumption}{Assumption}

\newlength{\widebarargwidth}
\newlength{\widebarargheight}
\newlength{\widebarargdepth}

\makeatletter
\long\def\@makecaption#1#2{
        \vskip 0.8ex
        \setbox\@tempboxa\hbox{\small {\bf #1:} #2}
        \parindent 1.5em  
        \dimen0=\hsize
        \advance\dimen0 by -3em
        \ifdim \wd\@tempboxa >\dimen0
                \hbox to \hsize{
                        \parindent 0em
                        \hfil
                        \parbox{\dimen0}{\def\baselinestretch{0.96}\small
                                {\bf #1.} #2
                                }
                        \hfil}
        \else \hbox to \hsize{\hfil \box\@tempboxa \hfil}
        \fi
        }
\makeatother

\long\def\comment#1{}
\definecolor{battleshipgrey}{rgb}{0.52, 0.52, 0.51}
\definecolor{darkgray}{rgb}{0.66, 0.66, 0.66}
\definecolor{darkgreen}{rgb}{0.0, 0.2, 0.13}
\definecolor{darkspringgreen}{rgb}{0.09, 0.45, 0.27}
\definecolor{dukeblue}{rgb}{0.0, 0.0, 0.61}
\definecolor{olivedrab7}{rgb}{0.24, 0.2, 0.12}
\definecolor{darkblue}{rgb}{0.0, 0.0, 0.55}
\definecolor{darkscarlet}{rgb}{0.34, 0.01, 0.1}
\definecolor{candyapplered}{rgb}{1.0, 0.03, 0.0}
\definecolor{ao(english)}{rgb}{0.0, 0.5, 0.0}
\definecolor{applegreen}{rgb}{0.55, 0.71, 0.0}

\makeatletter
\def\moverlay{\mathpalette\mov@rlay}
\def\mov@rlay#1#2{\leavevmode\vtop{%
   \baselineskip\z@skip \lineskiplimit-\maxdimen
   \ialign{\hfil$\m@th#1##$\hfil\cr#2\crcr}}}
\newcommand{\charfusion}[3][\mathord]{
    #1{\ifx#1\mathop\vphantom{#2}\fi
        \mathpalette\mov@rlay{#2\cr#3}
      }
    \ifx#1\mathop\expandafter\displaylimits\fi}
\makeatother

\DeclareMathOperator{\trace}{trace}

\newcommand{\dims}{\ensuremath{d}}
\newcommand{\real}{\ensuremath{\mathbb{R}}}

\newcommand{\naturalnum}{\ensuremath{\mathbb{N}}}

\newcommand{\Ind}{\ensuremath{\mathbb{I}}}
\newcommand{\borel}{\ensuremath{\mathcal{B}}}

\newcommand{\Exs}{\ensuremath{{\mathbb{E}}}}
\newcommand{\Prob}{\ensuremath{{\mathbb{P}}}}
\newcommand{\Law}{\mathcal{L}}
\newcommand{\Normal}{\ensuremath{\mathcal{N}}}

\DeclarePairedDelimiterX{\infdivx}[2]{(}{)}{%
  #1\;\delimsize\|\;#2%
}
\newcommand{\kldiv}{\mathrm{KL}\infdivx}
\newcommand{\tvdist}[2]{\ensuremath{ d_{\text{\tiny{TV}}}\parenth{#1, #2} }}

\newcommand{\Ot}{\widetilde{O}}

\newcommand{\brackets}[1]{\left[ #1 \right]}
\newcommand{\parenth}[1]{\left( #1 \right)}

\newcommand{\braces}[1]{\left\{ #1 \right \}}
\newcommand{\abss}[1]{\left| #1 \right |}
\newcommand{\angles}[1]{\left\langle #1 \right \rangle}
\newcommand{\ceils}[1]{\left\lceil #1 \right \rceil}

\newcommand{\tp}{^\top}

\newcommand{\matsnorm}[2]{\left|\!\left|\!\left| #1 \right|\!\right|\!\right|_{{#2}}}
\newcommand{\vecnorm}[2]{\left\| #1\right\|_{#2}}
\newcommand{\enorm}[1]{\vecnorm{#1}{2}}

\newcommand{\trltwo}{\text{trace}-\ell_2}

%% file: appendix_hmc_ssc.tex

\section{Additional proofs for the proposal overlap}
\label{app:add_proofs_proposal}

In order to show that for any $q_0\in\real^d$, the inverse map $q\mapsto\Go(q_0,q)$ of $p\mapsto\Fo(q_0,p)$ is well-defined, we need the following lemma.

\begin{lemma}[Hadamard, see also in \cite{gordon_diffeomorphisms_1972}]\label{lem:hadamard_reverse}
A $C^1$-map $f$ from $\real^n$ to $\real^n$ is a diffeomorphism if and only if f is
proper and the Jacobian $\det (\partial f_i/\partial x_j)$ never vanishes. 
\end{lemma}
$f$ is said to be proper if $f^{-1}(K)$ is compact for any compact set $K$. In Euclidean space, it is equivalent to require $f$ to be continuous and $|f(x)|\to\infty$ as  $|x|\to \infty$. Now we provide the proof of Lemma \ref{lem:J_F_p_bound_better}.

\begin{proof}[Proof of Lemma~\ref{lem:J_F_p_bound_better}]
  $\Fo_j$ satisfies
  \begin{align*}
    \Fo_j(q, p) &= q + j \step p - \frac{j\step^2}{2} \gradf(q) - \step^2 \sum_{\ell=1}^{{j-1}} (j-\ell) \gradf(\Fo_\ell(q, p)).
  \end{align*}
  Taking one derivative with respect to $q$, we obtain 
  \begin{align*}
      \D_1 \Fo_j(q,p)=\Ind_\dims -\frac{j\step^2}{2}\hessf_q-\step^2 \sum_{\ell=1}^{j-1} (j-\ell)\hessf_{\Fo_\ell(q,p)} \D_1 \Fo_\ell(q,p)
  \end{align*}
  Taking one derivative with respect to $p$, we obtain
  \begin{align*}
    \D_2 \Fo_j(q, p) &= j\step \Ind_\dims - \step^2 \sum_{\ell=1}^{j-1} (j-\ell) \hessf_{\Fo_\ell(q, p)} \D_2 \Fo_\ell(q, p).
  \end{align*}
  By induction, we conclude that
  \begin{align*}
    \enorm{\D_1 \Fo_j(q,p)-\Ind_\dims}\leq j^2 \step^2 \Lparam,\quad 
    \enorm{\D_2 \Fo_j(q, p) - j\step \Ind_\dims} \leq j^3  \step^3 \Lparam.
  \end{align*}

Given that we assumed $K\step L^{\frac12}\leq \frac14$, we have $\vecnorm{\D_2\Fo(q,p)-K\step\Ind_d}{2}\leq\frac1{16}{K\step}$. So the Jacobian of $p\mapsto \Fo(q,p)$ never vanishes for fixed $q\in\real^d$. By this bound on the Jacobian, we also have that $\vecnorm{\Fo(q,p)}{2}\to\infty$ as $\enorm{p}\to \infty$. 

As a result, Lemma \ref{lem:hadamard_reverse} ensures that for fixed $q\in\real^d$, the mapping $p\mapsto \Fo(q,p)$ is a bijection from $\real^n$ to $\real^n$. We denote its inverse by $\tilq \mapsto \Go(q,\tilq)$. Moreover, $\D_2\Go(q,\cdot)$ exists and is continuous.

For any $x,y\in\real^\dims$ we have $\Go(x, \Fo(x,y))=y$, take the derivative with respect to $y$
  \begin{align}
    \label{eq:JG_vs_JF}
    \D_2 \Go(x, \Fo(x, y)) \ \D_2 \Fo(x, y) = \Ind_\dims.
  \end{align}
  Thus, we have $\enorm{\D_2\Go(x,y)-\frac{1}{K\step}\Ind_d}\leq\frac{1}{15}\frac{1}{K\step}$. 
  Similarly, since 
  $\Fo(x, \Go(x, y)) = y$, take derivative with respect to $x$,
  \begin{align*}
      \D_1\Fo(x,\Go(x,y))+\D_2 \Fo(x,\Go(x,y))\D_1 \Go(x,y)=0,
  \end{align*}
so $\enorm{\D_1 \Go(x,y)}=\enorm{-\D_2 \Fo(x,\Go(x,y))^{-1}\D_1\Fo(x,\Go(x,y))}\leq \frac{16}{15K\step}\cdot\frac{17}{16}=\frac{17}{15K\step}$.
\end{proof}

\begin{proof}[Proof of Lemma~\ref{lem:Fj_q_p_bound}]
  For $1 \leq j \leq K$, we have
  \begin{align*}
    \Fo_j(q, p) = q + j\step p - \frac{j\step^2}{2} \gradf(\cq) - \step^2 \sum_{\ell=1}^{j-1} (j-\ell) \gradf(\Fo_\ell(q, p)). 
  \end{align*}
  We prove the result by induction. For $j=1$, we have
  \begin{align*}
    \enorm{\Fo_1(q, p) - \Fo_1(\tilq, \tilp)} &\leq \enorm{q - \tilq} + \step \enorm{p - \tilp} + \frac12 \step^2 \Lparam \enorm{q - \tilq} \\
    &\leq 2 \enorm{q - \tilq} + \step \enorm{p - \tilp}.
  \end{align*}
  Suppose the result is verified up to $j-1$, for $j$, we have
  \begin{align*}
    &\quad \enorm{\Fo_j(q, p) - \Fo_j(\tilq, \tilp)} \\
    & \leq \enorm{q - \tilq} + j \step \enorm{p - \tilp} + \frac12 j \step^2 \Lparam \enorm{q - \tilq} + \step^2\sum_{\ell=1}^{j-1} (j-\ell) \parenth{2 \enorm{q - \tilq} + 2 \ell \step \enorm{p - \tilp} } \\
    &\leq 2 \enorm{q - \tilq} + 2 j \step \enorm{p - \tilp}.
  \end{align*}
\end{proof}

Before proving Lemma \ref{lem:trace_logdet}, we need the following bound on the operator norms of the inverse matrices.

\begin{lemma}\label{lem:inverse_enorm_control}
    Assume $A\in\real^{d\times d}$ is an invertible matrix, and $\enorm{A-\Ind_d}\leq \alpha$ for some $\alpha\in (0,1)$, then we have
    \begin{equation*}
        \enorm{A^{-1}-\Ind_d}\leq \frac{\alpha}{1-\alpha}
    \end{equation*}
\end{lemma}
\begin{proof}[Proof of Lemma \ref{lem:inverse_enorm_control}]
By assumption we have 
\begin{equation*}
\begin{aligned}
    \alpha \enorm{A^{-1}-\Ind_d}&\geq \enorm{A-\Ind_d}\enorm{A^{-1}-\Ind_d}\\
    &\geq \enorm{(A-\Ind_d)(A^{-1}-\Ind_d)}\\
    &\geq \enorm{2\Ind_d-A-A^{-1}}\\
    &\geq \enorm{A^{-1}-\Ind_d}-\enorm{A-\Ind_d},
\end{aligned}
\end{equation*}
and the inequality follows.
\end{proof}
Now we can bound the log-determinant and trace term together in Lemma \ref{lem:trace_logdet}.
\begin{proof}[Proof of Lemma \ref{lem:trace_logdet}]
We define a real-valued function $f$ on $[0,1]$ by 
\begin{equation*}
    f(s)\defn -\log\det(\Ind_d-sA)-s\trace(A)-2s^2\trace\parenth{AA\tp}.
\end{equation*}
Take first-order derivative,
\begin{equation*}
    f'(s)=\trace\parenth{(\Ind_d-sA)^{-1}A}-\trace(A)-4s\trace\parenth{AA\tp},
\end{equation*}
We notice that $f'(0)=0$. Take second-order derivative, 
\begin{equation*}
\begin{aligned}
f''(s)&=\trace\parenth{(\Ind_d-sA)^{-1}A(\Ind_d-sA)^{-1} A}-4\trace\parenth{AA\tp}\\
&\leq \vecnorm{(\Ind_d-sA)^{-1}A}{F}^2- 4\vecnorm{A}{F}^2\\
&\leq \vecnorm{A}{F}^2\brackets{\enorm{(\Ind_d-sA)^{-1}}^2-4},
\end{aligned}
\end{equation*}
thus $f''(s)\leq 0$ for $s\in [0,1]$. Because by Lemma \ref{lem:inverse_enorm_control}, $\enorm{(\Ind-sA)^{-1}-\Ind_d}\leq 1$ when $\enorm{sA}\leq \frac12$. 

As a result, $f$ is a decreasing on $[0,1]$, and the inequality follows by $f(1)\leq f(0)=0$.
\end{proof}

\section{Additional proofs for the acceptance rate}
\label{app:add_proofs_acceptance}
In this section, we prove the results in Subsection~\ref{sub:intermediate_concentration_results}.

\subsection{Basic concentration bounds}
\label{sub:add_proofs_acceptance_basic_concentration}
\begin{proof}[Proof of Lemma~\ref{lem:grad_norm_bound}.]
  We have
\begin{align*}
  &\quad \Exs_{q \sim e^{-f}} \enorm{\gradf(q)}^{2\ell} \\
  &=  \Exs \vecnorm{\gradf(q)}{2}^{2(\ell-1)} \gradf(q)\tp \gradf(q) \\
  &\overset{(i)}{=} \Exs \enorm{\gradf(q)}^{2(\ell-1)} \trace\parenth{\hessf_q} + (\ell-1) \Exs \enorm{\gradf(q)}^{2(\ell-2)} \trace\parenth{2 \gradf(q) \tp \hessf_q \gradf(q)} \\
  &\leq \Exs \enorm{\gradf(q)}^{2(\ell-1)} \parenth{\trace\parenth{\hessf_q} + 2(\ell-1) \Lparam}.
\end{align*}
(i) follows from the multivariate integration by parts (or Green's identities, $\nabla u = \gradf(q) e^{-f(q)}$, $\nabla v = \vecnorm{\gradf(q)}{2}^{2(\ell-1)} \gradf(q) $), and the boundary term vanishes because $e^{-f}$ is subexponential-tailed. 
By H\"older's inequality, we obtain that
\begin{align*}
  \Exs_{q \sim e^{-f}} \enorm{\gradf(q)}^{2\ell} \leq (\trH + 2 (\ell-1) \Lparam)^\ell = \Cl^\ell.
\end{align*}
\end{proof}

\begin{proof}[Proof of Lemma~\ref{lem:pHp_bound}]
Given fixed $x$, we define $H\defn \parenth{\hessf_{x}}^2$ for ease of notation. We have
  \begin{align*}
    &\quad \Exs_{p \sim \Normal(0, \Ind_\dims)} \parenth{p \tp H p}^{\ell} \\
    &= \Exs_{p \sim \Normal(0, \Ind_\dims)} \parenth{p \tp H p}^{\ell-1} p \tp H p \\
    &\overset{(i)}{=} \Exs_{p \sim \Normal(0, \Ind_\dims)} \parenth{p \tp H p}^{\ell-1}  \trace\parenth{H} + 2(\ell-1) \Exs_{p \sim \Normal(0, \Ind_\dims)} \parenth{p \tp H p}^{\ell-2} p \tp  H^2 p \\
    &\overset{(ii)} \leq \Exs_{p \sim \Normal(0, \Ind_\dims)}  \parenth{p \tp  H p}^{\ell-1} \parenth{\trace\parenth{H} + 2 (\ell-1) \Lparam^2}\\
    &\leq \Exs_{p \sim \Normal(0, \Ind_\dims)}\parenth{p \tp  H p}^{\ell-1} \parenth{L^2 d+2(l-1)L^2}
  \end{align*}
  (i) follows from the multivariate integration by parts, where we used $\nabla u = p e^{-\frac{1}{2} \enorm{p}^2}$, $\nabla v = \parenth{p \tp H p}^{\ell-1} H p $, and the boundary term vanishes. Inequality $(ii)$ holds since $0\preceq H\preceq L^2 \Ind_{d}$. By H\"older's inequality, we obtain the desired result. 
\end{proof}


\begin{proof}[Proof of Lemma~\ref{lem:gradHp}]
  Given $q$, the random variable $\gradf(q) \tp \hessf_{q} p$ is Gaussian with variance $\enorm{\nabla^2 f_q \nabla f(q)}^2$. Therefore, from moments of the standard Gaussian, we know
  \begin{align*}
    \brackets{\Exs_{(q, p) \sim {e^{-f} \times \Normal(0, \Ind_\dims)}} \parenth{\gradf(q) \tp \hessf_{q} p}^\ell }^{\frac{1}{\ell}} &\leq \brackets{\Exs_{q \sim e^{-f}} \enorm{\nabla^2 f_q \gradf(q)}^{\ell} (\ell-1)!! }^{\frac{1}{\ell}}\\
    & \leq \brackets{\Exs_{q \sim e^{-f}}\enorm{\gradf(q)}^{\ell} \Lparam^\ell \ell^{\frac{\ell}{2}}}^{\frac{1}{\ell}}\\
    & \leq \ell^{\frac12} L\Cl^{\frac12}.
  \end{align*}
  The last step follows from Lemma~\ref{lem:grad_norm_bound}.
\end{proof}

\subsection{Third-order derivative bounds}
\label{sub:add_proofs_acceptance_third_order_bounds}
\begin{proof}[Proof of Lemma~\ref{lem:jerk_ppp_bound}]
  Define 
  \begin{align*}
    g(p,p,p) \defn \nabla^3 f_x[p,p,p].
  \end{align*}
  We have $\Exs g(p) = 0$. Applying Theorem 2.\textcolor{red}{16} in~\cite{adamczak2021moments} for Gaussian chaos, we obtain
  \begin{align}
    \parenth{\Exs g(p)^\ell}^{\frac{1}{\ell}} \leq c\sum_{k=1}^3 \sum_{J \subseteq [k]}\sum_{\mathcal P \in \mathcal P(J)} \ell^{|\mathcal P|/2} \vecnorm{\mathbb E\nabla^k g(p)}{\mathcal P},\label{eq:gaussianchaos}  
  \end{align}
  where $\mathcal P(J)$ is the set of partitions of $J$ and $c$ is a universal constant. For a partition $\mathcal P = \{I_1,\dots,I_r\}$ and $k$ dimensional multi-index tensor $(\mathbf{a}_{i_1\dots i_k})$, the $\|\mathbf a\|_{\mathcal P}$ is defined in an analogous fashion as Eq.~\eqref{eq:def_multi_index_norm}
  \begin{align*}
    \vecnorm{\mathbf{a}}{\mathcal P}^2 \defn \sup_x\sum_{\mathbf{i}_{[k] \backslash J}} \brackets{\sum_{\mathbf{i}_J}\mathbf a_{\mathbf{i}} \parenth{\prod_{\rho=1}^r x^{\rho}_{\mathbf{i}_{I_\rho}}} }^2,
\end{align*}
where the sup is over the set $\braces{\forall \rho \leq r, \sum_{\mathbf{i}_{I_\rho}}\parenth{x_{\mathbf{i}_{I_\rho}}^\rho }^2 \leq 1}$ with
\begin{align*}
    \sum_{\mathbf i_{I_r}} x_{\mathbf i_{I_r}}
\end{align*}

  We bound summation by separating into several cases. 
  \begin{itemize}
    \item $k = 3, J = \braces{1,2,3}$, then 
    \begin{align*}
      &\quad \sum_{\mathcal P \in \mathcal P(J)} \ell^{|\mathcal P|/2}\vecnorm{\Exs \nabla^k g(p)}{\mathcal P} \\
      &= \ell^{3/2}\vecnorm{3 \nabla^3 f_x}{\braces{1}\braces{2}\braces{3}} + 3 \ell \vecnorm{3 \nabla^3 f_x}{\braces{12}\braces{3}} + \ell^{1/2}\vecnorm{3 \nabla^3 f_x}{\braces{123}} \\
      &\leq 15 \ell^{3/2}\vecnorm{\nabla^3 f_x}{\braces{123}}.
    \end{align*}
    \item $k = 3, J = \braces{1,2}$ or $k = 3, J = \braces{1}$, the term is smaller than $c \ell \vecnorm{\nabla^3 f_x}{\braces{123}}$.
    \item $k = 2$, then $\Exs \nabla^2 g(p) = 0$. 
    \item $k = 1$, then 
    \begin{align*}
      \vecnorm{\Exs \nabla g(p)}{\braces{1}} = 3 \ell^{\frac12} \enorm{ \parenth{\partial_i \trace(\hessf)}_{i=1}^\dims } \leq  3 \ell^{\frac12} \dims^{\frac12} \vecnorm{\jerkf_x}{\braces{12}\braces{3}}.
    \end{align*}
  \end{itemize}
  Combining the above terms, we conclude.
\end{proof}

\begin{proof}[Proof of Lemma~\ref{lem:jerk_pp_norm_bound}]
  Let $A \defn \jerkf_{q}$ and $g(p) \defn \enorm{\jerkf_{q}[p, p, \cdot]}^2$. We have
  \begin{align*}
    &\quad g(p) =  \sum_k \parenth{ \sum_{ij} A_{ijk} p_i p_j \sum_{l m} A_{lmk} p_l p_m }.
  \end{align*}
  This is a fourth order (coupled) Gaussian chaos with tensor $B$ which has $ijlm$-th coefficient being $\sum_k A_{ijk} A_{lmk}$. Its derivatives are
  \begin{align*}
    &\Exs \nabla g (p) = 0 \\
    &\Exs \nabla^2 g(p) = \Big(4 \sum_{k=1}^n B_{ijkk} + 8\sum_{k=1}^n B_{ikjk}\Big)_{i,j=1}^n\numberthis\label{eq:secondder}\\
    &\Exs \nabla^3 g(p) = 0\\
    &\Exs \nabla^4 g(p) = 24 (B_{ijlm})_{i,j,l,m=1}^n.
  \end{align*}
    First, we bound the $\braces{1234}$ norm of the tensor $\frac{1}{24}\Exs \nabla^4 g(p) = B$. From the definition of $B$, we have
    \begin{align*}
       \vecnorm{B}{\braces{1234}}^2 &= \sum_{ijlm}B_{ijlm}^2 \\
       &= \sum_{ijlm} \parenth{\sum_k A_{ijk} A_{lmk}}^2 \\
       &\overset{(i)}{\leq} \sum_{ijlm} (\sum_kA_{ijk}^2)(\sum_lA_{lmk}^2) \\
       &=\vecnorm{A}{\braces{123}}^4.
    \end{align*}
    (i) follows from Cauchy-Schwarz inequality. 
  Next, for any partition $\mathcal{P}$ of $[4]$, we have that 
  \begin{align*}
    \vecnorm{B}{\mathcal{P}} \leq  \vecnorm{B}{\braces{1234}} \leq \vecnorm{A}{\braces{123}}^2.
  \end{align*}
  Therefore
  \begin{align}
    \label{eq:firstpart}
    \sum_{J \subseteq [4]}\sum_{\mathcal P \in \mathcal P(J)} \ell^{|\mathcal P|/2} \vecnorm{\Exs \nabla^4 g(p)}{\mathcal P}
      \lesssim \ell^{2}\vecnorm{B}{\braces{1234}} \leq \ell^{2}\vecnorm{A}{\braces{123}}^2.
  \end{align}
  
  Next, we bound the Frobenius norm of $\Exs \nabla^2 g(p)$. For the first term in Equation~\eqref{eq:secondder}, we have
  \begin{align*}
    \sum_{ij}(\sum_l B_{ijll})^2 &=
    \sum_{ij} (\sum_{lk}A_{llk}A_{ijk})^2\\
    &\leq \|A\|^2_{\braces{12}\braces{3}}\enorm{(\trace{A_{i\cdot\cdot}})_{i=1}^d}^2\\
    &\leq \dims \|A\|^4_{\braces{12}\braces{3}}.
  \end{align*}
  For the second part in~\eqref{eq:secondder}, we have
  \begin{align*}
    \sum_{ij}(\sum_l B_{iljl})^2
    &= \sum_{ij}(\sum_l \sum_k A_{ilk}A_{jlk})^2\\
    &\leq \sum_{ij}(\sum_{lk}A_{ilk}^2)(\sum_{lk}A_{jlk}^2)\\
    &\leq (\sum_{ilk} A_{ilk}^2)^2\\
    &=\|A\|_{\braces{123}}^4.
  \end{align*}
Hence
\begin{align}
    \vecnorm{\Exs \nabla^2 g(p)}{F}^2 \lesssim \dims^\frac12 \|A\|^2_{\braces{12}\braces{3}}+\|A\|_{\braces{123}}^2.\label{eq:frobeniusbound}
\end{align}
Therefore
\begin{align*}
  \sum_{J \subseteq [2]}\sum_{\mathcal P \in \mathcal P(J)} \ell^{|\mathcal P|/2} \vecnorm{\Exs \nabla^2 g(p)}{\mathcal P}
  &\lesssim \ell \parenth{ \dims^{\frac12} \|A\|^2_{\braces{12}\braces{3}}+\|A\|_{\braces{123}}^2 }.\numberthis\label{eq:secondpart}
\end{align*}
   From Equation~\eqref{eq:gaussianchaos} in the proof of Lemma~\ref{lem:jerk_ppp_bound} which is based on Theorem 2.15 in~\cite{adamczak2021moments}, and combining Equations~\eqref{eq:firstpart} and~\eqref{eq:secondpart}, we obtain
   \begin{align*}
    \matsnorm{g(p) - \Exs g(p)}{\ell}
    &\lesssim
     \ell^2\parenth{\dims^{\frac12}\|A\|^2_{\braces{12}\braces{3}}+\|A\|_{\braces{123}}^2}.
   \end{align*}
   Finally, we upper bound the expectation $\Exs g(p)$. Note that only the $B_{iijj}$ and $B_{ijij}$ terms survive in the expectation. Hence
   \begin{align*}
       \Exs g(p) &= \sum_{\braces{ij},\ i\neq j}2B_{iijj} + 4B_{ijij} + 3\sum_i B_{iiii}\\
       &=2\sum_{ij}\sum_{k}A_{ijk}A_{ijk} + \sum_{ij}\sum_{k}A_{jjk}A_{iik} + 3\sum_{i} \sum_k A^2_{iik} \\
       &=4\|A\|_{\braces{123}}^2 + 2\sum_k(\sum_{i} A_{iik})^2 + 3\sum_{ik} A^2_{iik}\\
       &\leq 4\|A\|_{\braces{123}}^2 
        + 2d\|A\|_{\braces{12}\braces{3}}^2 + 3\|A\|_{\braces{123}}^2\\
        &\lesssim \|A\|_{\braces{123}}^2 + d\vecnorm{A}{\braces{12}\braces{3}}^2.
   \end{align*}
   Therefore, overall
   \begin{align*}
       \matsnorm{g(p)}{\ell} &\lesssim 
       \ell^2\parenth{\dims^{\frac12} \|A\|^2_{\braces{12}\braces{3}}+\|A\|_{\braces{123}}^2} + \|A\|_{\braces{123}}^2 + d\vecnorm{A}{\braces{12}\braces{3}}^2 
   \end{align*}
\end{proof}

\begin{proof}[Proof of Lemma~\ref{lem:pHp_diff_bound}]
  We have
  \begin{align*}
    \cp_t\tp \hessf_{\cq_t} \cp_t - \cp_0 \tp \hessf_{\cq_0} \cp_0 = \int_0^t -2\gradf(\cq_s)\tp \hessf_{\cq_s}\cp_s + \jerkf_{\cq_s}[\cp_s, \cp_s, \cp_s] ds.
  \end{align*}
  Then
  \begin{align*}
    &\quad \Exs_{(\cq_0, \cp_0) \sim \target \times \Normal(0, \Ind_\dims)} \parenth{\cp_t\tp \hessf_{\cq_t} \cp_t - \cp_0 \tp \hessf_{\cq_0} \cp_0}^{\ell}  \\
    &\leq \Exs t^{\ell-1} \int_0^t \abss{ -2\gradf(\cq_s)\tp \hessf_{\cq_s}\cp_s + \jerkf_{\cq_s}[\cp_s, \cp_s, \cp_s]}^\ell ds \\
    &\leq t^{\ell-1} 2^{\ell-1} \brackets{\int_0^t \Exs 2^\ell \parenth{\gradf(\cq_s)\tp \hessf_{\cq_s}\cp_s}^\ell + \Exs \parenth{\jerkf_{\cq_s}[\cp_s, \cp_s, \cp_s]}^\ell  ds}  \\
    &\overset{(i)}{=} t^{\ell} 2^{\ell-1} \brackets{ 2^\ell \Exs \parenth{\gradf(\cq_0)\tp \hessf_{\cq_0}\cp_0}^\ell + \Exs \parenth{\jerkf_{\cq_0}[\cp_0, \cp_0, \cp_0]}^\ell} \\
    &\overset{(ii)}{\leq} \brackets{c t \parenth{\ell^{\frac12} \Lparam \Cl^{\frac12} +  \ell^{\frac32} \vecnorm{\jerkf_\cdot}{\braces{123}} + \ell^{\frac12} \dims^{\frac12} \vecnorm{\jerkf_\cdot}{\braces{12}\braces{3} } }}^{\ell}.
  \end{align*}
  (i) uses the fact that the continuous Hamiltonian dynamics preserves the Hamiltonian and $(\cq_s, \cp_s)$ has the same law as $(\cq_0, \cp_0)$. 
  (ii) follows from Lemma~\ref{lem:gradHp} and Lemma~\ref{lem:jerk_ppp_bound}.
\end{proof}

\begin{proof}[Proof of Lemma~\ref{lem:Hp_diff_bound}]
  We have
  \begin{align*}
    &\quad \hessf_{\cq_t} \cp_t - \hessf_{\cq_0} \cp_0 \\
    &= \int_0^t \jerkf_{\cq_s}[\cp_s, \cp_s, \cdot] - \hessf_{\cq_s} \gradf(\cq_s) ds.
  \end{align*}
  Hence
  \begin{align*}
    &\quad \Exs \enorm{\hessf_{\cq_t} \cp_t - \hessf_{\cq_0} \cp_0}^{2\ell} \\
    &\leq t^{2\ell-1} \int_0^t \Exs \enorm{\jerkf_{\cq_s}[\cp_s, \cp_s, \cdot] - \hessf_{\cq_s} \gradf(\cq_s)}^{2\ell} ds \\
    &\leq t^{2\ell-1} \int_0^t 2^{2\ell-1} \parenth{\Exs \enorm{\jerkf_{\cq_s}[\cp_s, \cp_s, \cdot]}^{2\ell} + \Lparam^{2\ell} \enorm{\gradf(\cq_s)}^{2\ell} } ds \\
    &\overset{(i)}{=} t^{2\ell} 2^{2\ell-1} \Exs \parenth{ \enorm{\jerkf_{\cq_0}[\cp_0, \cp_0, \cdot]}^{2\ell} + \Lparam^{2\ell} \enorm{\gradf(\cq_0)}^{2\ell} } \\
    &\overset{(ii)}{\leq} \brackets{c t \parenth{\vecnorm{\jerkf_\cdot}{\braces{123}} + \dims^{\frac12} \vecnorm{\jerkf_\cdot}{\braces{12}\braces{3}} + \Lparam\Cl^{\frac12} }}^{2\ell}.
  \end{align*}
  (i) uses the fact that the continuous Hamiltonian dynamics preserves the Hamiltionian and $(\cq_s, \cp_s)$ has the same law as $(\cq_0, \cp_0)$. 
  (ii) follows from Lemma~\ref{lem:grad_norm_bound} and Lemma~\ref{lem:jerk_pp_norm_bound}.
\end{proof}

\begin{proof}[Proof of Lemma~\ref{lem:qp_discrete_continuous_diff}]
  We have
  \begin{align*}
    \cq_t - q_t &= - \int_0^t \int_0^s \parenth{\gradf(\cq_\tau) - \gradf(\cq_0)} d\tau ds \\
    &= - \int_0^t \int_0^s \int_0^\tau \hessf_{\cq_\iota} \cp_\iota d\iota d\tau ds.
  \end{align*}
  Then we have
  \begin{align*}
    \matsnorm{\cq_t - q_t}{2\ell} 
    &\overset{(i)}{\leq} t^{3} \matsnorm{\hessf_{\cq_0} \cp_0}{2\ell} \\
    &\overset{(ii)}{\leq} t^{3} \Lparam \dl^{\frac12}.
  \end{align*}
  (i) follows from Eq.~\eqref{eq:random_vector_norm_integral} and the fact that $(\cq_s, \cp_s)$ has the same law as $(\cq_0, \cp_0)$. (ii) follows from Lemma~\ref{lem:pHp_bound}.
   For the other term, we have
  \begin{align*}
    \cp_\step - p_\step &= - \int_0^\step \gradf(\cq_t) - \frac12 \gradf(\cq_0) - \frac12 \gradf(\cq_\step) dt  -  \frac{\step}{2} \parenth{\gradf(\cq_\step) - \gradf(q_\step) } \\
    &= - \int_0^\step \parenth{\frac12 \step - t} \hessf_{\cq_t} \cp_t dt - \frac{\step}{2} \parenth{\gradf(\cq_\step) - \gradf(q_\step) }\\
    &= - \int_0^\step \parenth{\frac12 \step - t} \parenth{\hessf_{\cq_t} \cp_t -\hessf_{\cq_0} \cp_0} dt - \frac{\step}{2} \parenth{\gradf(\cq_\step) - \gradf(q_\step) }.
  \end{align*}
  Hence
  \begin{align*}
    \matsnorm{\cp_\step - p_\step}{2\ell} &\leq \frac{\step^2}{2} \sup_{t \in [0, \step]} \matsnorm{\hessf_{\cq_t} \cp_t -\hessf_{\cq_0} \cp_0}{2\ell} + \frac{\step\Lparam}{2} \matsnorm{\cq_\step - q_\step}{2\ell} \\
    &\overset{(i)}{\leq} c\step^3 \parenth{\vecnorm{\jerkf_\cdot}{\braces{123}} + \dims^{\frac12} \vecnorm{\jerkf_\cdot}{\braces{12}\braces{3}} + \Lparam^{\frac32}\dl^{\frac12}}
  \end{align*}
  (i) follows from Lemma~\ref{lem:Hp_diff_bound}, the $\cq_t - q_t$ bound above and $\step^2 \Lparam \leq 1$. 
\end{proof}

%% file: separableexamples.tex
\section{Proofs for examples of functions with small third order derivative tensor bounds}\label{app:separableexamples}

\begin{proof}[Proof of Lemma~\ref{lem:ridgeseparable}]
  Note that the second derivative and the third derivative of $f$ at $x$ are given by
  \begin{align}
    \label{eq:logisticthirdtensor}
    \hessf_\theta &= \sum_{i=1}^n u''_i(a_i \tp \theta) a_i \otimes a_i, \notag \\
    \jerkf_\theta &= \sum_{i=1}^n u'''_i(a_i \tp \theta) a_i \otimes a_i \otimes a_i.
  \end{align}
  First, for the operator norm of the Hessian, we have
  \begin{align*}
    \enorm{\hessf_\theta} &\leq \sum_{i=1}^n \abss{u''_i(a_i\tp \theta)} \big\|a_i^{2\otimes}\big\|\\
    &= \sum_{i=1}^n \abss{u''_i(a_i\tp \theta)} \enorm{a_i}^2\\
    &\leq n\rho_2.
  \end{align*}
Similarly, for the $\braces{1}\braces{2}\braces{3}$-norm of the third derivative, we have
\begin{align*}
    \vecnorm{\nabla^3 f_\theta}{\braces{1}\braces{2}\braces{3}} &\leq \sum_{i=1}^n \big|u'''_i(a_i^\top \theta)\big| \vecnorm{{a_i}^{3\otimes}}{\braces{1}\braces{2}\braces{3}}\\
    & \leq \sum_{i=1}^n |u'''_i(a_i^\top \theta)| \enorm{a_i}^3\\
    &\leq n \rho_3.
\end{align*}

To bound the $\|.\|_{\braces{123}}$ norm:
\begin{align*}
    \vecnorm{\jerkf_\theta}{\braces{123}} &\leq \sum_{i=1}^n \abss{u'''_i(a_i^\top \theta)}\vecnorm{a_i^{3\otimes}}{\braces{123}}\\
    & \leq \sum_{i=1}^n \abss{u'''_i(a_i^\top \theta)}\enorm{a_i}^3\\
    & \leq  n \rho_3.
\end{align*}

Finally note that the $\|.\|_{\braces{12}\braces{3}}$ norm is upper bounded by the $\|.\|_{\braces{123}}$ norm.

\end{proof}

\begin{proof}[Proof of Lemma~\ref{lem:logisticreg}]

Calculating the second derivative of the logistic loss function:
    \begin{align*}
        \frac{d^2}{dy^2}\ell(y,-1) = \frac{d^2}{dy^2}\ell(y,+1) = \frac{e^{-y}}{(1+e^{-y})^2} > 0.
    \end{align*}
    Therefore, $\mathcal L$ is convex in $\theta$. Additionally, for $f(\theta) \defn \mathcal L(\theta) + \frac{\alpha^2}{2} \enorm{\theta}^2$, $e^{-f}$ satisfies isoperimetric inequality with coefficient $\Omega(\alpha)$.
    Moreover, for smoothness of $\ell$, note that
    \begin{align*}
        \Big|\frac{d^2}{dy^2}\ell(y,-1)\Big| 
        \leq 1.
    \end{align*}
    Moreover,
    the third derivative of $\ell$ is given by
    \begin{align*}
        \frac{d^3}{dy^3}\ell(y,-1) = \frac{-e^{-y}(1+e^{-y})^2 - 2e^{-y}(1+e^{-y})(-e^{-y})}{(1+e^{-y})^4} &= \frac{e^{-3y} - e^{-y}}{(1+e^{-y})^4},
    \end{align*}
    which is bounded as
    \begin{align*}
        \big|\frac{d^3}{dy^3}\ell(y,-1)\big| \leq 1.
    \end{align*}
    Therefore, using Lemma~\ref{lem:ridgeseparable} we have
    \begin{align*}
        \|\nabla^2\mathcal L(\theta)\|_2 \leq n,
    \end{align*}
    which implies that $f$ is $n + \alpha$ smooth. Moreover, using Lemma~\ref{lem:ridgeseparable} for the third derivative
    \begin{align*}
       \|\jerkf(\theta)\|_{\braces{12}\braces{3}} \leq n.
    \end{align*}
    Plugging this into the mixing bound in Corollary~\ref{cor:ridgesep} for Ridge separable functions completes the proof.
\end{proof}

\begin{proof}[Proof of Lemma~\ref{lem:twolayernet}]

We note that $\theta_j\in\real^{d'}$ for each $j\in \{1,\ldots,m\}$. We write $\theta\in\real^{m d'}$ as their concatenation, that is, $\theta\tp =(\theta_1\tp\, \theta_2\tp\,\ldots\theta_m\tp)$. 

For any $\xi,\phi,\chi\in\real^{md'}$. The derivative of the loss function $\Law$ with respect to $\theta$ can be computed by

\begin{equation*}
\begin{aligned}
\nabla \Law(\theta)[\xi]&=\sum_{i=1}^n \braces{2(f_{NN}(x_i)-y_i)\sum_{j=1}^m w_j \sigma'(\theta_j\tp x_i)x_i\tp \xi_j}\\
\nabla^2 \Law(\theta)[\xi,\phi]&=2\sum_{i=1}^n \brackets{\sum_{k=1}^m w_k \sigma'(\theta_k\tp x_i)x_i\tp \phi_k} \brackets{\sum_{j=1}^m w_j \sigma'(\theta_j\tp x_i)x_i\tp \xi_j}+\\
&\qquad2\sum_{i=1}^n (f_{NN}(x_i)-y_i)\sum_{j=1}^m w_j \sigma''(\theta_j\tp x_i)x_i\tp \xi_j x_i\tp \phi_j 
\end{aligned}
\end{equation*}
As a result, for any $\xi \in \real^{md'}$, we have
\begin{equation*}
\begin{aligned}
\nabla^2 \Law(\theta)[\xi,\xi]&\leq 2\sum_{i=1}^n \brackets{\sum_{j=1}^m w_j \sigma'(\theta_j\tp x_i)x_i\tp \xi_j}^2+2\sum_{i=1}^n 2mc\sum_{j=1}^m \abss{w_j \sigma''(\theta_j\tp x_i)}(x_i\tp \xi_j)^2\\
&\overset{(i)}\leq 2\sum_{i=1}^n c^2\brackets{\sum_{j=1}^m \abss{x_i\tp \xi_j}}^2+2\sum_{i=1}^n 2mc^2\sum_{j=1}^m (x_i\tp \xi_j)^2\\
&\leq 2\sum_{i=1}^n c^2 m\sum_{j=1}^m (x_i\tp \xi_j)^2+2\sum_{i=1}^n 2mc^2\sum_{j=1}^m (x_i\tp \xi_j)^2\\
&\leq 6mc^2 \sum_{i=1}^n\sum_{j=1}^m (x_i\tp \xi_j)^2\overset{(ii)}\leq 6nmc^2\sum_{j=1}^m \enorm{\xi_j}^2= 6nmc^2  \enorm{\xi}^2
\end{aligned}
\end{equation*}
where we used $|w_j|\leq 1$ and $|\sigma'(z)|,|\sigma''(z)|\leq c$ in inequality $(i)$. Inequality $(ii)$ holds since $\enorm{x_i}=1$. Thus we have
\begin{equation*}
    \enorm{\nabla^2 \Law (\theta)}\lesssim mnc^2.
\end{equation*}

The third-order tensor $\nabla^3 \Law(\theta)$ can be computed: 
\begin{equation*}
\begin{aligned}
    \nabla^3 \Law (\theta)[\xi,\phi,\chi]&= 2\sum_{i=1}^n \brackets{\sum_{j=1}^m w_j \sigma'' (\theta_j\tp x_i)x_i\tp\phi_j x_i\tp \chi_j}\brackets{\sum_{j=1}^m w_j \sigma' (\theta_j\tp x_i)x_i\tp\xi_j}+\\
    &\quad 2\sum_{i=1}^n \brackets{\sum_{j=1}^m w_j \sigma'' (\theta_j\tp x_i)x_i\tp\xi_j x_i\tp \chi_j}\brackets{\sum_{j=1}^m w_j \sigma' (\theta_j\tp x_i)x_i\tp\phi_j}+\\
    &\quad 2\sum_{i=1}^n \brackets{\sum_{j=1}^m w_j \sigma'' (\theta_j\tp x_i)x_i\tp\xi_j x_i\tp \phi_j}\brackets{\sum_{j=1}^m w_j \sigma' (\theta_j\tp x_i)x_i\tp\chi_j}+\\
    &\quad2\sum_{i=1}^n (f_{NN}(x_i)-y_i)\sum_{j=1}^m w_j \sigma'''(\theta_j\tp x_i) (x_i\tp \xi_j) (x_i\tp \chi_j)(x_i\tp \phi_j).
\end{aligned}
\end{equation*}
To bound the $\braces{1}\braces{2}\braces{3}$-norm of $\nabla^3 \Law(\theta)$, it is adequate to maximize $\abss{\nabla^3\Law (\theta)[\xi,\xi,\xi]}$ for $\enorm{\xi}\leq 1$ due to symmetry (A proof is provided in \cite{zhang_best_2012}). We have 
\begin{equation*}
\begin{aligned}
  \nabla^3 \Law (\theta)[\xi,\xi,\xi]&= 6\sum_{i=1}^n \brackets{\sum_{j=1}^m w_j \sigma'' (\theta_j\tp x_i)(x_i\tp \xi_j)^2}\brackets{\sum_{j=1}^m w_j \sigma' (\theta_j\tp x_i)x_i\tp\xi_j}+\\
  &\quad 2\sum_{i=1}^n (f_{NN}(x_i)-y_i)\sum_{j=1}^m w_j \sigma'''(\theta_j\tp x_i) (x_i\tp \xi_j)^3,
\end{aligned}
\end{equation*}
thus we have
\begin{equation*}
\begin{aligned}
\abss{\nabla^3 \Law (\theta)[\xi,\xi,\xi]}&\lesssim \sum_{i=1}^n \brackets{\sum_{j=1}^m c (x_i\tp \xi_j)^2}\brackets{\sum_{j=1}^m c \abss{x_i\tp \xi_j}} + \sum_{i=1}^n mc \sum_{j=1}^m
c \abss{x_i\tp\xi_j}^3\\
&\overset{(i)}\lesssim nc^2\brackets{\sum_{j=1}^m  \enorm{\xi_j}^2}\brackets{\sum_{j=1}^m  \enorm{ \xi_j}} + n mc^2 \sum_{j=1}^m \enorm{\xi_j}^3\\
&\lesssim nc^2\enorm{\xi}^2\brackets{\sum_{j=1}^m  \enorm{ \xi_j}} + n mc^2 \enorm{\xi}\sum_{j=1}^m \enorm{\xi_j}^2\\
&\overset{(ii)}\lesssim nc^2\sqrt{m}\enorm{\xi}^3 + n mc^2 \enorm{\xi}^3 \lesssim mnc^2 \enorm{\xi}^3,
\end{aligned}
\end{equation*}
where inequality $(i)$ holds since $\enorm{x_i}= 1$, and inequality $(ii)$ holds by Cauchy-Schwartz inequality. To bound the $\braces{123}$-norm, we define third-order tensors $A$ and $B$ by
\begin{equation}
\begin{aligned}
A[\xi,\phi,\chi]\defn 2\sum_{i=1}^n \brackets{\sum_{j=1}^m w_j \sigma'' (\theta_j\tp x_i)x_i\tp\phi_j x_i\tp \chi_j}\brackets{\sum_{j=1}^m w_j \sigma' (\theta_j\tp x_i)x_i\tp\xi_j}\\
B[\xi,\phi,\chi]\defn 2\sum_{i=1}^n (f_{NN}(x_i)-y_i)\sum_{j=1}^m w_j \sigma'''(\theta_j\tp x_i) (x_i\tp \xi_j) (x_i\tp \chi_j)(x_i\tp \phi_j).
\end{aligned}
\end{equation}
We note that
\begin{equation*}
\nabla^3 \Law (\theta)[\xi,\phi,\chi]= A[\xi,\phi,\chi]+A[\chi,\xi,\phi]+A[\phi,\chi,\xi]+B[\xi,\phi,\chi],
\end{equation*}
so we have $\matsnorm{\nabla^3 \Law(\theta)}{\braces{123}}\leq 3\matsnorm{A}{\braces{123}}+  \matsnorm{B}{\braces{123}}$, and we control $A$ and $B$ separately.

We use $e_{p}\in\real^{md'}$ to denote the unit vector with $p$-th entry to be $1$, and all other entries being $0$. $e_{pj}\in\real^{d'}$ denotes the $j$-th component of $e_p$. Then we have
\begin{equation}
\begin{aligned}
\matsnorm{A}{\braces{123}}^2 &=\sum_{p,q,\ell=1}^{md'} A[e_p,e_q,e_\ell]^2  \\
&=4\sum_{p,q,\ell=1}^{md'}\braces{\sum_{i=1}^n \brackets{\sum_{j=1}^m w_j \sigma'' (\theta_j\tp x_i)x_i\tp e_{pj} x_i\tp e_{qj}}\brackets{\sum_{j=1}^m w_j \sigma' (\theta_j\tp x_i)x_i\tp e_{\ell j}}}^2\\
&\overset{(i)}\leq 4\sum_{p,q,\ell=1}^{md'}\sum_{i=1}^n \brackets{\sum_{j=1}^m w_j \sigma'' (\theta_j\tp x_i)x_i\tp e_{pj} x_i\tp e_{qj}}^2\brackets{\sum_{j=1}^m w_j \sigma' (\theta_j\tp x_i)x_i\tp e_{\ell j}}^2\\
&\overset{(ii)}\leq 4\sum_{p,q,\ell=1}^{md'}\sum_{i=1}^n c^4\brackets{\sum_{j=1}^m \abss{x_i\tp e_{pj} x_i\tp e_{qj}}}^2\brackets{\sum_{j=1}^m \abss{x_i\tp e_{\ell j}}}^2\\
&\overset{(iii)}\leq 4\sum_{p,q,\ell=1}^{md'}\sum_{i=1}^n c^4 m\brackets{\sum_{j=1}^m \abss{x_i\tp e_{pj} }^2}\brackets{\sum_{j=1}^m \abss{x_i\tp e_{\ell j}}^2}
\brackets{\sum_{j=1}^m \abss{x_i\tp e_{q j}}^2}\\
\end{aligned}
\end{equation}
where inequalities $(i),(iii)$ follows from Cauchy-Schawrtz inequality, and inequality $(ii)$ follows since $|w_j|\leq 1$ and $|\sigma'(z)|,|\sigma''(z)|\leq c$. Then $\matsnorm{A}{\braces{123}}$ can be controlled by 
\begin{equation}\label{eq:A_123_norm_bound}
\begin{aligned}
\matsnorm{A}{\braces{123}}^2&\leq 4c^4 m\sum_{i=1}^n\sum_{p,q,\ell=1}^{md'} \brackets{\sum_{j=1}^m \abss{x_i\tp e_{pj} }^2}\brackets{\sum_{j=1}^m \abss{x_i\tp e_{\ell j}}^2}
\brackets{\sum_{j=1}^m \abss{x_i\tp e_{q j}}^2}\\
&\leq 4c^4 m\sum_{i=1}^n \brackets{\sum_{j=1}^m \sum_{p=1}^{md'}\abss{x_i\tp e_{pj} }^2}\brackets{\sum_{j=1}^m\sum_{q=1}^{md'} \abss{x_i\tp e_{\ell j}}^2}
\brackets{\sum_{j=1}^m \sum_{\ell=1}^{md'}\abss{x_i\tp e_{q j}}^2}\\
&\overset{(i)}=  4c^4 m\sum_{i=1}^n \brackets{\sum_{j=1}^m \enorm{x_i }^2}\brackets{\sum_{j=1}^m \enorm{x_i}^2}
\brackets{\sum_{j=1}^m \enorm{x_i}^2}\overset{(ii)}\leq 4nm^4c^4.
\end{aligned}
\end{equation}
where equality $(i)$ holds since for each $j\in \braces{1,\ldots,m}$, $x_i\tp e_{pj}$ is non-zero  only for $p$ ranging over the $j$-th group, thus $\sum_{p=1}^{md'} (x_i\tp e_{pj})^2=\enorm{x_i}^2$. Inequality $(ii)$ holds since $\enorm{x_i}=1$.  

Similarly, $\matsnorm{B}{\braces{123}}$ can be controlled by 
\begin{equation*}
\begin{aligned}
\matsnorm{B}{\braces{123}}^2 &=\sum_{p,q,\ell=1}^{md'} B[e_p,e_q,e_\ell]^2 \\
&= \sum_{p,q,\ell=1}^{md'}\braces{2\sum_{i=1}^n (f_{NN}(x_i)-y_i)\sum_{j=1}^m w_j \sigma'''(\theta_j\tp x_i) (x_i\tp e_{pj}) (x_i\tp e_{qj})(x_i\tp e_{\ell j})}^2\\
&\overset{(i)}\leq  4\sum_{p,q,\ell=1}^{md'}4m^2 c^4\braces{\sum_{i=1}^n \sum_{j=1}^m 
\abss{(x_i\tp e_{pj}) (x_i\tp e_{qj})(x_i\tp e_{\ell j})}}^2\\
&\overset{(ii)}\leq  4\sum_{p,q,\ell=1}^{md'}4m^2 c^4nm\sum_{i=1}^n \sum_{j=1}^m 
(x_i\tp e_{pj})^2 (x_i\tp e_{qj})^2(x_i\tp e_{\ell j})^2\\
\end{aligned}
\end{equation*}
where inequality $(i)$ holds since $\abss{w_j}\leq 1$, $\abss{\sigma'''(z)}\leq 1$, $|f_{NN}(x_i)|\leq mc$ and $|y_i|\leq mc$. Inequality $(ii)$ holds by applying Cauchy-Schwartz inequality twice. Thus $\matsnorm{B}{\braces{123}}$ can be bounded by
\begin{equation}\label{eq:B_123_norm}
\begin{aligned}
\matsnorm{B}{\braces{123}}^2 &\leq 16 nm^3 c^4 \sum_{i=1}^n\sum_{j=1}^m\brackets{\sum_{p=1}^{md'}(x_i\tp e_{pj})^2} 
\brackets{\sum_{q=1}^{md'}(x_i\tp e_{qj})^2}\brackets{\sum_{\ell=1}^{md'}(x_i\tp e_{\ell j})^2}\\
&\leq 16 nm^3 c^4 \sum_{i=1}^n\sum_{j=1}^m \enorm{x_i}^2\cdot \enorm{x_i}^2 \cdot \enorm{x_i}^2\\
&=16 n^2 m^4 c^4,
\end{aligned}
\end{equation}
where the last equality holds since $\enorm{x_i}=1$. Combining Eq.\eqref{eq:A_123_norm_bound} and Eq.\eqref{eq:B_123_norm}, we have
\begin{equation*}
\matsnorm{\nabla^3 \Law (\theta)}{\braces{123}}\lesssim nm^2 c^2 +\sqrt{n}m^2 c^2 \lesssim nm^2 c^2
\end{equation*}
\end{proof}

%% file: ref.bib
@article{bou2020coupling,
	author = {Bou-Rabee, Nawaf and Eberle, Andreas and Zimmer, Raphael},
	journal = {The Annals of applied probability},
	number = {3},
	pages = {1209--1250},
	publisher = {JSTOR},
	title = {Coupling and convergence for Hamiltonian monte carlo},
	volume = {30},
	year = {2020}}

@article{mangoubi2017rapid,
	author = {Mangoubi, Oren and Smith, Aaron},
	journal = {arXiv preprint arXiv:1708.07114},
	title = {Rapid mixing of Hamiltonian Monte Carlo on strongly log-concave distributions},
	year = {2017}}

@article{gozlan2015dimension,
	author = {Gozlan, Nathael and Roberto, Cyril and Samson, Paul-Marie},
	journal = {Calculus of Variations and Partial Differential Equations},
	number = {3-4},
	pages = {899--925},
	publisher = {Springer},
	title = {From dimension free concentration to the {P}oincar{\'e} inequality},
	volume = {52},
	year = {2015}}

@article{gromov1983topological,
	author = {Gromov, Mikhael and Milman, Vitali D},
	journal = {American Journal of Mathematics},
	number = {4},
	pages = {843--854},
	publisher = {JSTOR},
	title = {A topological application of the isoperimetric inequality},
	volume = {105},
	year = {1983}}

@article{milman2009role,
	author = {Milman, Emanuel},
	journal = {Inventiones mathematicae},
	number = {1},
	pages = {1--43},
	publisher = {Springer},
	title = {On the role of convexity in isoperimetry, spectral gap and concentration},
	volume = {177},
	year = {2009}}

@incollection{cheeger2015lower,
	author = {Cheeger, Jeff},
	booktitle = {Problems in analysis},
	pages = {195--200},
	publisher = {Princeton University Press},
	title = {A lower bound for the smallest eigenvalue of the {L}aplacian},
	year = {2015}}

@inproceedings{maz1960classes,
	author = {Maz'ya, Vladimir Gilelevich},
	booktitle = {Doklady Akademii Nauk},
	organization = {Russian Academy of Sciences},
	pages = {527--530},
	title = {Classes of domains and imbedding theorems for function spaces},
	volume = {133},
	year = {1960}}

@article{salvatier2016probabilistic,
	author = {Salvatier, John and Wiecki, Thomas V and Fonnesbeck, Christopher},
	journal = {PeerJ Computer Science},
	pages = {e55},
	publisher = {PeerJ Inc.},
	title = {Probabilistic programming in Python using PyMC3},
	volume = {2},
	year = {2016}}

@article{gouraud_hmc_2025,
	title = {{HMC} and {Underdamped} {Langevin} {United} in the {Unadjusted} {Convex} {Smooth} {Case}},
	volume = {13},
	issn = {2166-2525},
	url = {https://epubs.siam.org/doi/10.1137/23M1608963},
	doi = {10.1137/23M1608963},
	language = {en},
	number = {1},
	urldate = {2026-02-11},
	journal = {SIAM/ASA Journal on Uncertainty Quantification},
	author = {Gouraud, Nicolaï and Bris, Pierre Le and Majka, Adrien and Monmarché, Pierre},
	month = mar,
	year = {2025},
	pages = {278--303}
}

@inproceedings{altschuler_resolving_2023,
	series = {Proceedings of {Machine} {Learning} {Research}},
	title = {Resolving the {Mixing} {Time} of the {Langevin} {Algorithm} to its {Stationary} {Distribution} for {Log}-{Concave} {Sampling}},
	volume = {195},
	booktitle = {Proceedings of {Thirty} {Sixth} {Conference} on {Learning} {Theory}},
	publisher = {PMLR},
	author = {Altschuler, Jason and Talwar, Kunal},
	editor = {Neu, Gergely and Rosasco, Lorenzo},
	month = jul,
	year = {2023},
	pages = {2509--2510}
}

@article{chen_optimal_2022,
	title = {Optimal {Convergence} {Rate} of {Hamiltonian} {Monte} {Carlo} for {Strongly} {Logconcave} {Distributions}},
	volume = {18},
	issn = {1557-2862},
	url = {https://theoryofcomputing.org/articles/v018a009},
	doi = {10.4086/toc.2022.v018a009},
	language = {en},
	number = {1},
	urldate = {2026-02-11},
	journal = {Theory of Computing},
	author = {Chen, Zongchen and Vempala, Santosh S.},
	year = {2022},
	pages = {1--18},
}

@article{vempala2019rapid,
	author = {Vempala, Santosh and Wibisono, Andre},
	journal = {Advances in neural information processing systems},
	title = {Rapid convergence of the unadjusted Langevin algorithm: Isoperimetry suffices},
	volume = {32},
	year = {2019}}

@article{dalalyan2017theoretical,
	author = {Dalalyan, Arnak S},
	journal = {Journal of the Royal Statistical Society: Series B (Statistical Methodology)},
	number = {3},
	pages = {651--676},
	publisher = {Wiley Online Library},
	title = {Theoretical guarantees for approximate sampling from smooth and log-concave densities},
	volume = {79},
	year = {2017}}

@inproceedings{lee2021structured,
	author = {Lee, Yin Tat and Shen, Ruoqi and Tian, Kevin},
	booktitle = {Conference on Learning Theory},
	organization = {PMLR},
	pages = {2993--3050},
	title = {Structured logconcave sampling with a restricted gaussian oracle},
	year = {2021}}

@article{carpenter2017stan,
	author = {Carpenter, Bob and Gelman, Andrew and Hoffman, Matthew D and Lee, Daniel and Goodrich, Ben and Betancourt, Michael and Brubaker, Marcus A and Guo, Jiqiang and Li, Peter and Riddell, Allen},
	journal = {Grantee Submission},
	number = {1},
	pages = {1--32},
	publisher = {ERIC},
	title = {Stan: a probabilistic programming language.},
	volume = {76},
	year = {2017}}

@article{grenander1994representations,
	author = {Grenander, Ulf and Miller, Michael I},
	journal = {Journal of the Royal Statistical Society: Series B (Methodological)},
	number = {4},
	pages = {549--581},
	publisher = {Wiley Online Library},
	title = {Representations of knowledge in complex systems},
	volume = {56},
	year = {1994}}

@article{lovasz1993random,
	author = {Lov{\'a}sz, L{\'a}szl{\'o} and Simonovits, Mikl{\'o}s},
	journal = {Random structures \& algorithms},
	number = {4},
	pages = {359--412},
	publisher = {Wiley Online Library},
	title = {Random walks in a convex body and an improved volume algorithm},
	volume = {4},
	year = {1993}}

@article{roberts1998optimal,
	author = {Roberts, Gareth O and Rosenthal, Jeffrey S},
	journal = {Journal of the Royal Statistical Society: Series B (Statistical Methodology)},
	number = {1},
	pages = {255--268},
	publisher = {Wiley Online Library},
	title = {Optimal scaling of discrete approximations to {Langevin} diffusions},
	volume = {60},
	year = {1998}}

@inproceedings{chewi2021optimal,
	author = {Chewi, Sinho and Lu, Chen and Ahn, Kwangjun and Cheng, Xiang and Le Gouic, Thibaut and Rigollet, Philippe},
	booktitle = {Conference on Learning Theory},
	organization = {PMLR},
	pages = {1260--1300},
	title = {Optimal dimension dependence of the {Metropolis}-{Adjusted} {Langevin} {Algorithm}},
	year = {2021}}

@article{bou2013nonasymptotic,
	author = {Bou-Rabee, Nawaf and Hairer, Martin},
	journal = {IMA Journal of Numerical Analysis},
	number = {1},
	pages = {80--110},
	publisher = {OUP},
	title = {Nonasymptotic mixing of the {MALA} algorithm},
	volume = {33},
	year = {2013}}

@article{durmus2017nonasymptotic,
	author = {Durmus, Alain and Moulines, Eric},
	journal = {The Annals of Applied Probability},
	number = {3},
	pages = {1551--1587},
	publisher = {Institute of Mathematical Statistics},
	title = {Nonasymptotic convergence analysis for the unadjusted {Langevin} algorithm},
	volume = {27},
	year = {2017}}

@article{neal2011mcmc,
	author = {Neal, Radford M},
	journal = {Handbook of {Markov} {Chain} {Monte} {Carlo}},
	number = {11},
	pages = {2},
	title = {{MCMC} using {Hamiltonian} dynamics},
	volume = {2},
	year = {2011}}

@article{lee2021lower,
	author = {Lee, Yin Tat and Shen, Ruoqi and Tian, Kevin},
	journal = {Advances in Neural Information Processing Systems},
	pages = {18812--18824},
	title = {Lower bounds on Metropolized sampling methods for well-conditioned distributions},
	volume = {34},
	year = {2021}}

@inproceedings{lee2020logsmooth,
	author = {Lee, Yin Tat and Shen, Ruoqi and Tian, Kevin},
	booktitle = {Conference on Learning Theory},
	organization = {PMLR},
	pages = {2565--2597},
	title = {Logsmooth Gradient Concentration and Tighter Runtimes for {Metropolized} {Hamiltonian} {Monte} {Carlo}},
	year = {2020}}

@article{dwivedi2018log,
	author = {Raaz Dwivedi and Yuansi Chen and Martin J. Wainwright and Bin Yu},
	journal = {Journal of Machine Learning Research},
	number = {183},
	pages = {1-42},
	title = {Log-concave sampling: {M}etropolis-{H}astings algorithms are fast},
	volume = {20},
	year = {2019}}

@article{mou2019high,
	author = {Mou, Wenlong and Ma, Yi-An and Wainwright, Martin J and Bartlett, Peter L and Jordan, Michael I},
	journal = {Journal of Machine Learning Research},
	pages = {42--1},
	title = {High-Order {L}angevin Diffusion Yields an Accelerated {MCMC} Algorithm.},
	volume = {22},
	year = {2021}}

@article{durmus2019high,
	author = {Durmus, Alain and Moulines, Eric},
	journal = {Bernoulli},
	number = {4A},
	pages = {2854--2882},
	publisher = {Bernoulli Society for Mathematical Statistics and Probability},
	title = {High-dimensional {Bayesian} inference via the unadjusted {Langevin} algorithm},
	volume = {25},
	year = {2019}}

@article{roberts1996geometric,
	author = {Roberts, Gareth O and Tweedie, Richard L},
	journal = {Biometrika},
	number = {1},
	pages = {95--110},
	publisher = {Oxford University Press},
	title = {Geometric convergence and central limit theorems for multidimensional {Hastings} and {Metropolis} algorithms},
	volume = {83},
	year = {1996}}

@article{chen2020fast,
	author = {Chen, Yuansi and Dwivedi, Raaz and Wainwright, Martin J and Yu, Bin},
	journal = {Journal of Machine Learning Research},
	number = {92},
	pages = {1--71},
	title = {Fast mixing of {Metropolized} {Hamiltonian} {Monte} {Carlo}: Benefits of multi-step gradients},
	volume = {21},
	year = {2020}}

@article{parisi1981correlation,
	author = {Parisi, Giorgio},
	journal = {Nuclear Physics B},
	number = {3},
	pages = {378--384},
	publisher = {Elsevier},
	title = {Correlation functions and computer simulations},
	volume = {180},
	year = {1981}}

@article{cheng2018convergence,
	author = {Cheng, Xiang and Bartlett, Peter},
	journal = {Proceedings of Machine Learning Research, Volume 83: Algorithmic Learning Theory},
	pages = {186--211},
	publisher = {Proceedings of Machine Learning Research},
	title = {Convergence of {Langevin} {MCMC} in {KL}-divergence},
	year = {2018}}

@article{besag1994comments,
	author = {Besag, JE},
	journal = {J. Roy. Statist. Soc. Ser. B},
	pages = {591--592},
	title = {Comments on ``{Representations} of knowledge in complex systems'' by {U.} {Grenander} and {MI} {Miller}},
	volume = {56},
	year = {1994}}

@article{sinclair1989approximate,
	author = {Sinclair, Alistair and Jerrum, Mark},
	journal = {Information and Computation},
	number = {1},
	pages = {93--133},
	publisher = {Elsevier},
	title = {Approximate counting, uniform generation and rapidly mixing {Markov} chains},
	volume = {82},
	year = {1989}}

@article{beskos2013optimal,
	author = {Beskos, Alexandros and Pillai, Natesh and Roberts, Gareth and Sanz-Serna, Jesus-Maria and Stuart, Andrew},
	journal = {Bernoulli},
	number = {5A},
	pages = {1501--1534},
	title = {Optimal tuning of the hybrid {M}onte {C}arlo algorithm},
	volume = {19},
	year = {2013}}

@article{betancourt2014optimizing,
	author = {Betancourt, MJ and Byrne, Simon and Girolami, Mark},
	journal = {arXiv preprint arXiv:1411.6669},
	title = {Optimizing the integrator step size for {H}amiltonian {M}onte {C}arlo},
	year = {2014}}

@article{creutz1988global,
	author = {Creutz, Michael},
	journal = {Physical Review D},
	number = {4},
	pages = {1228},
	publisher = {APS},
	title = {Global {M}onte {C}arlo algorithms for many-fermion systems},
	volume = {38},
	year = {1988}}

@article{abadi2016tensorflow,
	author = {Abadi, Mart{\'\i}n and Agarwal, Ashish and Barham, Paul and Brevdo, Eugene and Chen, Zhifeng and Citro, Craig and Corrado, Greg S and Davis, Andy and Dean, Jeffrey and Devin, Matthieu and others},
	journal = {arXiv preprint arXiv:1603.04467},
	title = {Tensorflow: Large-scale machine learning on heterogeneous distributed systems},
	year = {2016}}

@article{neal1994improved,
	author = {Neal, Radford M},
	journal = {Journal of Computational Physics},
	number = {1},
	pages = {194--203},
	publisher = {Elsevier},
	title = {An improved acceptance procedure for the hybrid {M}onte {C}arlo algorithm},
	volume = {111},
	year = {1994}}

@article{duane1987hybrid,
	author = {Duane, Simon and Kennedy, Anthony D and Pendleton, Brian J and Roweth, Duncan},
	journal = {Physics letters B},
	number = {2},
	pages = {216--222},
	publisher = {Elsevier},
	title = {Hybrid {M}onte {C}arlo},
	volume = {195},
	year = {1987}}

@article{alder1959studies,
	author = {Alder, Berni J and Wainwright, Thomas Everett},
	journal = {The Journal of Chemical Physics},
	number = {2},
	pages = {459--466},
	publisher = {American Institute of Physics},
	title = {Studies in molecular dynamics. I. General method},
	volume = {31},
	year = {1959}}

@book{vershynin2018high,
	author = {Vershynin, Roman},
	publisher = {Cambridge university press},
	title = {High-dimensional probability: An introduction with applications in data science},
	volume = {47},
	year = {2018}}

@article{zajkowski2019norms,
	author = {Zajkowski, Krzysztof},
	journal = {Statistics \& Probability Letters},
	pages = {147--152},
	publisher = {Elsevier},
	title = {Norms of sub-exponential random vectors},
	volume = {152},
	year = {2019}}

@article{lovasz1999hit,
	author = {Lov{\'a}sz, L{\'a}szl{\'o}},
	journal = {Mathematical programming},
	number = {3},
	pages = {443--461},
	publisher = {Springer},
	title = {Hit-and-run mixes fast},
	volume = {86},
	year = {1999}}

@article{adamczak2021moments,
	author = {Adamczak, Rados{\l}aw and Lata{\l}a, Rafa{\l} and Meller, Rafa{\l}},
	journal = {Electronic Journal of Probability},
	pages = {1--36},
	publisher = {Institute of Mathematical Statistics and Bernoulli Society},
	title = {Moments of {G}aussian chaoses in {B}anach spaces},
	volume = {26},
	year = {2021}}

@article{latala2006estimates,
	author = {Lata{\l}a, Rafa{\l}},
	journal = {The Annals of Probability},
	number = {6},
	pages = {2315--2331},
	publisher = {Institute of Mathematical Statistics},
	title = {Estimates of moments and tails of {G}aussian chaoses},
	volume = {34},
	year = {2006}}

@article{wu2021minimax,
	author = {Wu, Keru and Schmidler, Scott and Chen, Yuansi},
	journal = {The Journal of Machine Learning Research},
	number = {1},
	pages = {12348--12410},
	publisher = {JMLRORG},
	title = {Minimax mixing time of the Metropolis-adjusted Langevin algorithm for log-concave sampling},
	volume = {23},
	year = {2022}}

@inproceedings{laddha2020strong,
	author = {Laddha, Aditi and Lee, Yin Tat and Vempala, Santosh},
	booktitle = {Proceedings of the 52nd annual ACM SIGACT symposium on theory of computing},
	pages = {1212--1222},
	title = {Strong self-concordance and sampling},
	year = {2020}}

@article{zhang2022mean,
	author = {Zhang, Jingwei and Huang, Xunpeng},
	journal = {arXiv preprint arXiv:2205.09860},
	title = {Mean-Field Analysis of Two-Layer Neural Networks: Global Optimality with Linear Convergence Rates},
	year = {2022}}

@article{cao2020complexity,
	author = {Cao, Yu and Lu, Jianfeng and Wang, Lihan},
	journal = {Communications in Mathematical Sciences},
	number = {7},
	pages = {1827--1853},
	publisher = {International Press of Boston},
	title = {Complexity of randomized algorithms for underdamped {L}angevin dynamics},
	volume = {19},
	year = {2021}}

@article{shen2019randomized,
	author = {Shen, Ruoqi and Lee, Yin Tat},
	journal = {Advances in Neural Information Processing Systems},
	title = {The randomized midpoint method for log-concave sampling},
	volume = {32},
	year = {2019}}

@article{gordon_diffeomorphisms_1972,
	author = {Gordon, W. B.},
	doi = {10.1080/00029890.1972.11993118},
	issn = {0002-9890, 1930-0972},
	journal = {The American Mathematical Monthly},
	language = {en},
	month = aug,
	number = {7},
	pages = {755--759},
	title = {On the {Diffeomorphisms} of {Euclidean} {Space}},
	urldate = {2025-11-21},
	volume = {79},
	year = {1972}}

@article{bou-rabee_mixing_2023,
	author = {Bou-Rabee, Nawaf and Eberle, Andreas},
	doi = {10.3150/21-BEJ1450},
	issn = {1350-7265},
	journal = {Bernoulli},
	month = feb,
	number = {1},
	title = {Mixing time guarantees for unadjusted {Hamiltonian} {Monte} {Carlo}},
	url = {https://projecteuclid.org/journals/bernoulli/volume-29/issue-1/Mixing-time-guarantees-for-unadjusted-Hamiltonian-Monte-Carlo/10.3150/21-BEJ1450.full},
	urldate = {2025-11-27},
	volume = {29},
	year = {2023}}

@article{zhang_best_2012,
	title = {The {Best} {Rank}-1 {Approximation} of a {Symmetric} {Tensor} and {Related} {Spherical} {Optimization} {Problems}},
	volume = {33},
	issn = {0895-4798, 1095-7162},
	url = {http://epubs.siam.org/doi/10.1137/110835335},
	doi = {10.1137/110835335},
	language = {en},
	number = {3},
	urldate = {2026-02-09},
	journal = {SIAM Journal on Matrix Analysis and Applications},
	author = {Zhang, Xinzhen and Ling, Chen and Qi, Liqun},
	month = jan,
	year = {2012},
	pages = {806--821}
}

@article{bou-rabee_unadjusted_2025,
	title = {Unadjusted {Hamiltonian} {MCMC} with stratified {Monte} {Carlo} time integration},
	volume = {35},
	issn = {1050-5164},
	url = {https://projecteuclid.org/journals/annals-of-applied-probability/volume-35/issue-1/Unadjusted-Hamiltonian-MCMC-with-stratified-Monte-Carlo-time-integration/10.1214/24-AAP2116.full},
	doi = {10.1214/24-AAP2116},
	number = {1},
	urldate = {2026-02-11},
	journal = {The Annals of Applied Probability},
	author = {Bou-Rabee, Nawaf and Marsden, Milo},
	month = feb,
	year = {2025}
}

@article{altschuler_faster_2024,
	title = {Faster {High}-accuracy {Log}-concave {Sampling} via {Algorithmic} {Warm} {Starts}},
	volume = {71},
	issn = {0004-5411, 1557-735X},
	url = {https://dl.acm.org/doi/10.1145/3653446},
	doi = {10.1145/3653446},
	language = {en},
	number = {3},
	urldate = {2026-02-11},
	journal = {Journal of the ACM},
	author = {Altschuler, Jason M. and Chewi, Sinho},
	month = jun,
	year = {2024},
	pages = {1--55}
}

@article{
chizat2022meanfield,
title={Mean-Field Langevin Dynamics : Exponential Convergence and Annealing},
author={L{\'e}na{\"\i}c Chizat},
journal={Transactions on Machine Learning Research},
issn={2835-8856},
year={2022},
url={https://openreview.net/forum?id=BDqzLH1gEm},
note={}
}
